\documentclass[%
 reprint,
 amsmath,amssymb,
 aps,
]{revtex4-2}

\usepackage{graphicx}% Include figure files
\usepackage{dcolumn}% Align table columns on decimal point
\usepackage{bm}% bold math
\usepackage{physics}
\usepackage[caption=false]{subfig}

\usepackage{amsthm}

\newtheorem{theorem}{Theorem}[section]

\newtheorem{observation}[theorem]{Observation}
\newtheorem{definition}{Definition}[section]
\newtheorem{property}{Property}

\usepackage[dvipsnames]{xcolor}
\newcommand{\e}{\operatorname{e}}

\usepackage{hyperref}
\usepackage{algorithm}
\usepackage[noend]{algpseudocode}

\begin{document}

\preprint{APS/123-QED}

\title{Measurement-based quantum machine learning}

\author{Luis Mantilla Calder\'on\textsuperscript{1,2,3}}
\email{luis@cs.toronto.edu}
\altaffiliation{this work was conducted while the author was at the University of British Columbia.}

\author{Robert Raussendorf\textsuperscript{3,4,5}}
%\email{robert.raussendorf@itp.uni-hannover.de}

\author{Polina Feldmann\textsuperscript{3,4,6}}
%\email{polina.feldmann@ubc.ca}
\altaffiliation{These authors contributed equally to this work.}

\author{Dmytro Bondarenko\textsuperscript{3,4,7}}
\altaffiliation{These authors contributed equally to this work.}
%\email{dimbond@live.com}

\affiliation{\textsuperscript{1}Department of Computer Science, University of Toronto, Canada}
\affiliation{\textsuperscript{2}Vector Institute for Artificial Intelligence, Toronto, Canada}
\affiliation{\textsuperscript{3}Department of Physics \& Astronomy, University of British Columbia, Vancouver, Canada}
\affiliation{\textsuperscript{4}Stewart Blusson Quantum Matter Institute, University of British Columbia, Vancouver, Canada}
\affiliation{\textsuperscript{5}Institut für Theoretische Physik, Leibniz Universität Hannover, Germany}
\affiliation{\textsuperscript{6}Department of Electrical and Computer Engineering, University of British Columbia, Vancouver, Canada}
\affiliation{\textsuperscript{7}Department of Chemistry, University of British Columbia, Vancouver, Canada}

% \date{}

\begin{abstract}

Quantum machine learning (QML) leverages quantum computing for classical inference, furnishes the processing of quantum data with machine-learning methods, and provides quantum algorithms adapted to noisy devices. Typically, QML proposals are framed in terms of the circuit model of quantum computation. The alternative measurement-based quantum computing (MBQC) paradigm can exhibit lower circuit depths, is naturally compatible with classical co-processing of mid-circuit measurements, and offers a promising avenue towards error correction. Despite significant progress on MBQC devices, QML in terms of MBQC has been hardly explored. We propose the multiple-triangle ansatz (MuTA), a universal quantum neural network assembled from MBQC neurons featuring bias engineering, monotonic expressivity, tunable entanglement, and scalable training. We numerically demonstrate that MuTA can learn a universal set of gates in the presence of noise, a quantum-state classifier, as well as a quantum instrument, and classify classical data using a quantum kernel tailored to MuTA. Finally, we incorporate hardware constraints imposed by photonic Gottesman-Kitaev-Preskill qubits. Our framework lays the foundation for versatile quantum neural networks native to MBQC, allowing to explore MBQC-specific algorithmic advantages and QML on MBQC devices.

\end{abstract}

\maketitle

\section{\label{sec:intro}Introduction}

The idea of quantum computing has led to the development of several algorithms that bring significant speedups over their classical counterparts. Famous examples include Shor's algorithm, Grover's algorithm, and quantum simulation, which outperform classical algorithms in the large-scale limit. Quantum computing can also be used to process quantum data, with sources spanning quantum sensors, quantum communication devices, and quantum simulators. 
However, most quantum algorithms require a large number of qubits and gates with the addition of quantum error correction, which is a big challenge for current quantum hardware.  Thus, it is important to develop algorithms that utilize fewer resources and can tolerate more noise.

Measurement-based quantum computing (MBQC)~\cite{raussendorf_one-way_2001,raussendorf_measurement-based_2003,raussendorf_review_2012} executes a sequence of adaptive single-qubit measurements on an entangled resource state. 
%In certain settings, MBQC can be less demanding or more resource efficient than the standard unitary-circuit model of quantum computing. 
One of the best-known advantages over the standard unitary-circuit model of quantum computing is MBQC's compatibility with probabilistic entangling gates~\cite{duan2005efficient}. This makes MBQC particularly appealing for photonic implementations~\cite{chen2024heralded,lib2024resource,yokoyama2013ultra}. Cluster resource states can also be efficiently generated using ultracold atoms in optical lattices~\cite{mandel2003controlled}. Interestingly, measurement-based algorithms can exhibit a lower time complexity than their circuit-based counterparts~\cite{broadbent2009parallelizing}. For example, the ground state of the error-correcting toric code can be obtained in constant time~\cite{raussendorf_long-range_2005}. While MBQC usually incurs an (at most polynomial) overhead in qubit number, sequential generation of the resource state allows for balancing time complexity benefits against qubit overhead~\cite{raussendorf_measurement-based_2003}. MBQC supports topological error correction~\cite{raussendorf2006fault} and is inherently compatible with mid-circuit measurements with feedback.
Such measurements are known to reduce the qubit count in several algorithms, including Shor's algorithm~\cite{rossi_using_2021}. Impressively, for matrix product states these savings become exponential \cite{malz_preparation_2023}.

Recently, quantum machine learning (QML) \cite{biamonte_quantum_2017} has been explored as a method for developing novel quantum algorithms. QML uses quantum hardware
%, capable of implementing parametrized gates natively, 
as machine learning models to perform inference on datasets. The primary objective is to leverage quantum advantage, ideally even before reaching fault-tolerant systems. Recent literature, including Refs.~\cite{abbas_power_2021,cong_quantum_2019,beer_training_2020,schuld_quantum_2019,tilly_variational_2022}, has provided evidence for the utility of quantum neural networks (QNNs). Quantum speed-ups have been proven for variational quantum classifiers~\cite{jager2023universal} and quantum kernels~\cite{liu2021rigorous,jager2023universal}.
Beyond its potential for a computational advantage over classical approaches, QML is inherently valuable, offering automated development of quantum algorithms, novel approaches to quantum and classical problem solving, and anticipated applications in the analysis of quantum data.

In this paper, we introduce a framework for designing QML models based on the MBQC paradigm of quantum computing. In particular, we propose a family of feed-forward QNNs defined through parameterized single-qubit measurements on specifically structured graph states. We name our QNN architecture the multiple-triangle ansatz (MuTA). MBQC on graph states has previously been explored in the context of the variational quantum eigensolver, using edge decorations that yield a non-deterministic ansatz \cite{ferguson_measurement-based_2021}, or directly translating from the circuit model to MBQC~\cite{ferguson_measurement-based_2021, qin2024applicability}. In contrast, we propose general-purpose QNNs designed for deterministic MBQC. Reference~\cite{majumder2024variational} exploits the randomness inherent in MBQC for generative modeling using cluster states. Our ansatz is compatible with such partial randomness while offering greater flexibility, particularly in terms of connectivity. 
MuTA stands out by a unique combination of useful features, including universality, tunable entanglement, monotonicity of expressivity, versatile bias engineering, and scalable training. These attributes make the MuTA architecture a robust foundation for developing future QNN models within MBQC.

The paper is organized as follows. In Section \ref{sec: MBQC}, we briefly review and combine the concepts of MBQC and QML.  Section \ref{sec:MUTA} presents our main result. We introduce MuTA, a QNN architecture based on MBQC, and establish several desirable properties such as universality and versatile bias engineering. In Section \ref{sec:Experiments} we numerically demonstrate the practical applicability of MuTA by learning a universal set of gates, a quantum-state classifier, and a quantum instrument. We provide numerical evidence that MuTA tolerates various types of noise. Finally, we design a quantum kernel based on MuTA and use it to classify classical data. In Section~\ref{sec:hea}, we accommodate hardware restrictions arising in typical photonic implementations and suggest heuristic algorithms for training the adapted models. We close with a summary and outlook in Section~\ref{sec:discussion}.

\section{\label{sec: MBQC}Measurement-based quantum machine learning}
% \com{How about a section about machine learning in general?}

MBQC works by preparing an entangled state known as a \emph{resource state} and then performing single-qubit measurements that consume this state \cite{raussendorf_one-way_2001}. Examples of resource states include graph states,
\begin{equation}
\label{eq: graphstate}
    \ket{G} = \prod_{(i,j) \in E} \operatorname{CZ}_{ij} \ket{+}^{\otimes n}.
\end{equation}
Here, $G=(V,E)$ denotes a graph with vertices $V$ and edges $E$, $X\ket{+}=\ket{+}$, and CZ is a Controlled-$Z$ gate. $X$, $Y$, and $Z$ are Pauli operators. To use $\ket{G}$ for MBQC, we define sets of input and output qubits $I,O\subset V$, respectively, together with single-qubit measurements $M^i$ at each node $i\notin O$. A common choice, which we assume in this paper, is $M^i_{\alpha}=\cos(\alpha) X_i + \sin(\alpha) Y_i$, $\alpha \in (-\pi, \pi]$. However, extended MBQC allows $M^i$ to be in the $(X,Z)$ and $(Y,Z)$ planes, too ~\cite{browne_generalized_2007}. The input of the circuit can be a quantum state $\ket{\psi}_{\mathrm{in}}$, and in this case we replace the $\ket{+}$ qubits in $I$ by $\ket{\psi}_{\mathrm{in}}$ before applying the CZ gates in Eq.~\eqref{eq: graphstate}. If the input is classical, we can encode it in an input quantum state or in the measurement angles $\alpha(i)$ of $M_\alpha^i$. The output of the circuit is stored in $O$. 

To obtain a deterministic output in spite of the inherent randomness of quantum measurements, the measurement angles have to be conditioned on previous measurement outcomes. For MBQC on a graph state, the presence of \emph{flow} guarantees a viable conditioning~\cite{danos_determinism_2006}. Let $N_G(i)$, $i\in V$ denote the neighbors of $i$ in $G$. A resource state characterized by $(G, I, O)$ has flow iff there exists a function $f:V\setminus O \rightarrow V\setminus I$ and a partial order $<$ such that $\forall i \in V$
\begin{itemize}
    \item $(i, f(i)) \in E$, 
    \item $i < f(i)$,
    \item $\forall j \in N_G [f(i)] \setminus \{i\}: \quad i<j$.
\end{itemize}
With flow, measuring in a time order compatible with $<$ allows to compensate a $(-1)$-outcome on any qubit $i\notin O$ by subsequently applying $X_{f(i)}\prod_{j \in N_G [f(i)] \setminus \{i\}}Z_j$. This compensation, which ensures determinism, is equivalent to a simple adaptation of measurement angles and a conditional Pauli operator acting on the output state. To efficiently find the flow of a graph, we can use the algorithms presented in Refs.~\cite{de_beaudrap_finding_2008, mhalla_finding_2008}. There is a more general condition known as Pauli flow \cite{browne_generalized_2007}, which is necessary and sufficient for uniformly deterministic computation; however, for the sake of simplicity, we consider only graphs with flow.

Classical (supervised) machine learning aims to infer a map $\mathcal{M}$ from features $\mathcal{X}$ to labels $\mathcal{Y}$ based on a limited dataset $\mathcal{D} = \{(x_i, y_i)\}_{i=1}^N$ with $x_i\in\mathcal{X}$, $y_i=\mathcal{M}(x_i)\in\mathcal{Y}$. During training, the machine-learning model $f_\theta: \mathcal{X} \rightarrow \mathcal{Y}$, parameterized by $\theta$, learns to approximate $\mathcal{M}$ on $\mathcal{D}$. This learning process is driven by a loss function $\ell(\theta; \mathcal{D})$ which evaluates the difference between predicted labels $f_\theta(x_i)$ and true labels $y_i$. Minimizing $\ell(\theta; \mathcal{D})$ optimizes the variational parameter $\theta$. A key requirement for successful machine learning is that the model generalizes well to unseen (but probable) data.

Integrating the discussed concepts, we introduce a machine-learning model for measurement-based quantum machine learning (MB-QML) which is specified by the underlying graph state $\ket{G}$ with flow. The sets of mixed states $\mathcal{R}_{I/O}$ associated with the input and output qubits define the quantum features $\mathcal{X}=\mathcal{R}_I$ and labels $\mathcal{Y}=\mathcal{R}_O$. The model $f_\theta$ is the unitary embedding $\mathcal{U}_{\vec{\alpha}}:\mathcal{R}_I\to\mathcal{R}_O$ realized by an MBQC on $\ket{G}$ with measurement angles $\vec{\alpha}=(\alpha_0,\ldots,\alpha_{|V\setminus O|-1})$.
The latter constitute the variational parameter $\theta=\vec{\alpha}$ of our ansatz $\mathcal{U}_{\vec{\alpha}}$. Throughout, we assume that measurement angles are adapted for determinism according to the flow of $\ket{G}$. For pure input and output states and a dataset $\mathcal{D}=\{(\ket{\psi}_i, \ket{\phi}_i)\}_{i=1}^N$, a natural choice for the loss function is
\begin{equation}\label{eq:loss}
    \ell(\vec{\alpha},\mathcal{D})=1-\frac{1}{N}\sum_iF\big(\mathcal{U}_{\vec{\alpha}}\!\ket{\psi}_i,\ket{\phi}_i\big),
\end{equation}
where $F(\ket{\psi},\ket{\phi})=|\langle\psi|\phi\rangle|^2$ is the fidelity.

\section{The multiple-triangle ansatz \label{sec:MUTA}}

Classical deep neural networks (NNs) have shown an incredible power for learning functions of real-world data. Each node of a NN represents a small set of operations comprising an \emph{artificial neuron}. A connection from node $i$ to $j$ implies that the output of neuron $i$ serves as one of the inputs to neuron $j$. Several architectures have been proposed to solve problems with a different structure, including multi-layer perceptrons, convolutional neural networks, transformers, graph neural networks, and many more \cite{goodfellow2016deep, zhang2023dive}. Importantly, classical NNs are capable of (I) \emph{universally approximating} certain classes of functions, 
(II) \emph{tunable message passing} between different nodes, 
(III) \emph{engineerable bias}, and  (IV) \emph{learning} at scale. Networks composed of ReLU neurons---currently the most popular choice~\cite{nair2010rectified, goodfellow2016deep}---additionally exhibit (V) \emph{monotonicity of expressivity}. For further details on classical NNs and the expressivity of ReLU neurons see Appendix~\ref{App:ExpressMononton}. A quest for QNN models has driven the community to propose many ideas for constructing circuits that behave similarly to classical NNs~\cite{beer_training_2020, da_silva_quantum_2016, benedetti_parameterized_2019}.

\begin{figure*}[t] 
    \centering
    \subfloat[]{%
        \includegraphics[width=0.6\textwidth]{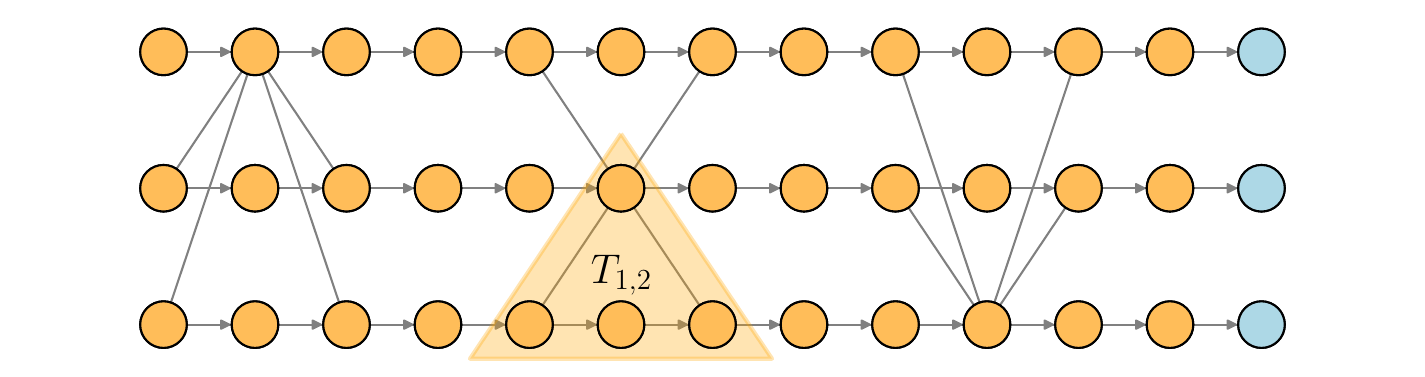}
        \label{fig:muta}
    }
    % \hspace{0.5cm}
    \subfloat[]{%
        \includegraphics[width=0.34\textwidth]{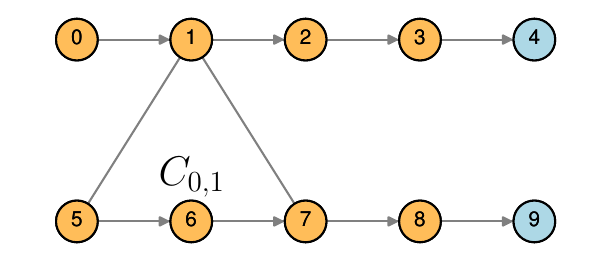}
        \label{fig:oneLayerMuta}
    }
    \caption{Illustrations of the multiple-triangle ansatz (MuTA) architecture: (a) A MuTA for three qubits, showing input qubits $I$ on the left, output qubits $O$ on the right (in blue), with arrows indicating the flow of $(G,I,O)$; (b) A $(2,0)$-MuTA layer.}
    \label{fig:muta_combined}
\end{figure*}

Here, we introduce QNNs which are native to MBQC and exhibit a number of beneficial properties similar to those present in classical NNs. The centerpiece of our framework is a resource state constructed from simple, yet universal, building blocks inspired by Ref.~\cite{danos_parsimonious_2005}. Our \emph{multiple-triangle ansatz (MuTA)} consists of 1D cluster states connected by multiple triangles as shown in Fig.~\ref{fig:muta_combined}. To navigate the qubits of this ansatz, we introduce the following sets: $Q_{i,k}$ denotes the qubit in row $i$ and column $k$, \(T_{i,j,k}=Q_{i,k}\cup Q_{j,k-1}\cup Q_{j,k}\cup Q_{j,k+1}\) is the set of qubits in a triangle connecting wires $i$ (tip) and $j$ (base) around column $k$, and \(C_{i,j,k}=Q_{j,k} \subset T_{i,j,k}\) forms the center of its base. We denote $C_{i,j,k}$ and $T_{i,j,k}$ by $C_{i,j}$ and $T_{i,j}$ whenever $k$ is unambiguous. To build MuTA layer $(n,i)$, we connect $n$ wires of length $5$ by a collection of triangles $T_{i,j,1}$ with $j\in J\subseteq \{0,\dots, n-1\}\backslash{\{i\}}$ (note that we count starting from zero). We call $(i,J)$ the intra-layer connectivity and say that a MuTA layer is fully connected if $J=\{0,\ldots,n-1\}\setminus\{i\}$. In the opposite case of no intra-layer connectivity, i.\,e., $J=\emptyset$, we denote the MuTA layer by $(n)$. A MuTA of depth $d$ consists of $d$ MuTA layers. 
We concatenate two layers $(n,i)$ and $(m,j)$ by identifying $k\leq \min(n,m)$ output qubits $K_O\subseteq O_l$ of layer $l$ with $k$ input qubits $K_I\subseteq I_{l+1}$ of layer $l+1$ via a map $\iota_{l,l+1}:K_O\rightarrow K_I$. The input qubits of the subsequent layer (output qubits of the preceding layer) that do not participate in the concatenation join the input (output) set of the overall graph state. The map $\iota_{l,l+1}$ determines the inter-layer connectivity. We show that this architecture satisfies the following desiderata:

\begin{property}[Determinism]
\label{prop:determinism}
Any measurement pattern on MuTA can be performed deterministically.
\end{property}
\begin{proof}

MuTA has a flow $f$ defined by \( f(Q_{i,k}) = Q_{i, k + 1} \) for each node \( i \in V\setminus O \). A compatible partial order $<$ is established as follows: 
$Q_{i,k} < Q_{i',k'}$ if (I) $k'=k+1$ and $i'=i$, (II) $k=k'=4m$, $i'\neq i$, and $Q_{i,k+1}$ is a tip qubit, or (III) $k=k'=4m+1$, $i'\neq i$, and $Q_{i',k'}$ is a tip qubit $(m\in\mathbb{N}_0)$.
Thus, for any measurement pattern, Theorem 1 from Ref.~\cite{danos_determinism_2006} explicitly states how to adapt measurement angles to obtain a deterministic outcome. 
\end{proof}

\begin{property}[Universality]
\label{prop:universality}
Given any $n$-qubit unitary operator \( V \in U(2^n) \), there exists a measurement pattern on a sufficiently deep MuTA of width $n$ that exactly implements $V$.

\end{property}
\begin{proof}
The single-qubit gates $V'\in U(2)$ and an entangling gate such as $\text{Ising}X\!X(\varphi)=\exp(-i \varphi X\otimes X/2)$ for $\varphi \notin \pi \mathbb{Z}$ form a universal set of gates \cite{nielsen2010quantum, bremner_practical_2002}, i.\,e., any \( V \in U(2^n) \) can be decomposed into a product of these gates. Both types of gates can be implemented by MuTA using the measurement patterns in Table~\ref{tab:muta-gates}. First, any single-qubit gate (including the identity) can be realized by measuring a five-qubit wire with angles given by the Euler decomposition $V'=R_x (-\lambda) R_z (-\phi) R_x (-\theta)$, where $R_x (\alpha)= \exp(-i\alpha X/2)$ and $R_z (\alpha)= \exp(-i\alpha Z/2)$~\cite{raussendorf_one-way_2001}. Second, the $\text{Ising}X\!X(\varphi)$ gate can be implemented on a MuTA layer $(2,0)$, which follows from the unitary gate corresponding to $C_{0,1}$ (see Appendix~\ref{append:Generators}, Eq.~\eqref{eq:unitarymuta20}). Therefore, a MuTA layer $(n)$ consisting of only $n$ disconnected five-qubit wires can apply arbitrary single-qubit unitaries to all wires, while a MuTA layer $(n,i)$ with a single triangle $T_{i,j}$ can implement $\text{Ising}X\!X(\varphi)$ on wires $i,j$. Concatenating these two types of layers with their corresponding measurement patterns yields a MuTA that exactly implements $V$.

\begin{table}[ht]
\centering
\renewcommand{\arraystretch}{1.5} 
\begin{tabular}{|c|c|}
\hline
\textbf{Gate} & \textbf{Measurement Pattern} \\
\hline
$R_x(-\lambda) R_z(-\phi) R_x (-\theta) \otimes I$  & 
$M_{0}^{8} M_{\lambda}^{3} M_{0}^{7} M_{\phi}^{2} M_{\theta}^{1} M_{0}^{6} M_{0}^{5} M_{0}^{0}$, \\
\hline
$\text{Ising}X\!X(-\varphi)$ & $
M_{0}^{8} M_{0}^{3} M_{0}^{7} M_{0}^{2} M_{0}^{1} M_{\varphi}^{6} M_{0}^{5} M_{0}^{0}$ \\
\hline
\end{tabular}
\caption{Measurement patterns for universal gates in MuTA. Indices are matched with those from Fig.~\ref{fig:oneLayerMuta}.}
\label{tab:muta-gates}
\end{table}

\end{proof}

\begin{property}[Tunable entanglement]
\label{prop:tunentanglement}
A MuTA layer $(n, i)$ with connectivity $(i,J)$ implements
\begin{enumerate}
    \item a non-entangling gate between wires $i$ and $j\in J$ if $C_{i,j}$ is measured at $\alpha_j\in\{0,\pi\}$,
    \item an entangling gate between wires $i$ and $j\in J$ if $C_{i,j}$ is measured at $\alpha_j\notin \{0,\pi\}$.
\end{enumerate}
Here, we call a gate $G$ entangling for wires $i$ and $j$ iff there exists some $n$-qubit state $\ket{\psi}=\ket{\psi_i}\otimes\ket{\psi_j}\otimes\ket{\psi_{\mathrm{rest}}}$ such that $G\ket{\psi}$ reduced to the $(i,j)$-subspace is entangled.
\end{property}

\begin{proof}

In Appendix~\ref{append:Generators}, see Eq.~\eqref{eq:unitarymuta} and the subsequent paragraph, we translate a general MuTA layer into a unitary circuit. This circuit consists of one- and two-qubit gates that correspond to individual MuTA measurements. The two-qubit unitaries belong to $C_{i,k}$ qubits and are of the form $U_{i,k}=\exp(i\alpha_kX_iX_k/2)$, where $\alpha_k$ denotes the measurement angle at $C_{i,k}$ (cf. Property~\ref{prop:universality}). All $U_{i,k}$ act simultaneously (cf. Property~\ref{prop:determinism}), implementing $G=\exp(iX_i\sum_{k\in J}\alpha_k X_k/2)$. Since single-qubit unitaries do not matter for the question of entanglement generation, it is sufficient to show that Property~\ref{prop:tunentanglement} holds for $G$.

Without loss of generality, we set $i=0$ and $j=1$. For $\alpha_1=0$, $G$ does not act on qubit 1; for $\alpha_1=\pi$ it acquires the form $G=X_1\otimes \tilde{G}$. In both cases, $G$ cannot entangle qubit $j$ with any other qubit, which proves the first statement.

For the second statement, we pick $\ket{\psi_{\mathrm{rest}}}=\ket{+}^{\otimes{n-2}}$. This turns $G\ket{\psi}$ reduced to the subspace of qubits $i=0$, $j=1$ into the pure state $G'(\ket{\psi_0}\otimes\ket{\psi_1})$ with $G'=\exp[iX_0(\alpha_1X_1+\gamma)/2]$ and $\gamma=\sum_{k\in J\setminus\{1\}} \alpha_k$. Since 1-qubit operations do not matter, it suffices to show that $G''=\exp[i\alpha_1X_0X_1/2]=\text{Ising}X\!X({-}\alpha_1)$, $\alpha_{1}\notin\{0,\pi\}$ is an entangling gate. To confirm this, we consider $\ket{\phi}=\text{Ising}X\!X({-}\alpha_1)\ket{0,0}$ and observe that the reduced density matrix $\rho=\Tr_1\ket{\phi}\!\bra{\phi}$ describes a mixed state on qubit 0, $\Tr\rho^2=(1+\cos^2\alpha_1)/2<1$ $\forall$ $\alpha_{1}\notin\{0,\pi\}$.

\end{proof}

\begin{property}[Bias engineering]\label{prop:bias}
The set of transformations implementable by a given MuTA depends on the chosen geometry---i.\,e., the overall depth, the width of individual layers, the inter-layer connectivity, and the intra-layer connectivity. It can be further constraint by restricting the values of some or all variational parameters, i.\,e., the measurement angles.
\end{property}
\begin{proof} Obvious.
\end{proof}

\begin{property}[Scalability]\label{prop:scalability}
A MuTA of depth $d$ and maximal width $n$ has $\leq 4 d n$ variational parameters.
\end{property}
\begin{proof} MuTA has at most as many variational parameters as qubits.
%,i.\,e., $4\times d\times n$.
Denoting the width of layer $l$ by $n_l\leq n$, this amounts to $4\sum_{l=1}^dn_l\leq 4 d n$.
\end{proof}

\begin{property}[Monotonicity of expressivity]\label{prop:monotonexpressivity}
    The following expansions of a given MuTA monotonically increase the set of implementable gates:
    \begin{enumerate}
        \item Adding layers.
        \item Adding disconnected wires.
        \item Adding admissible triangles to a layer $l$ of width $n$ and inserting an additional layer $l'$ of width $n$ between layers $l$ and $l+1$. Layer $l'$ can have any intralayer connectivity. It takes all outputs of layer $l$ as inputs. The interlayer connectivity between layers $l'$ and $l+1$ is inherited from the initial connectivity between $l$ and $l+1$.
    \end{enumerate}
\end{property}
\begin{proof}
    For weak monotonicity, it is sufficient to show that the expansions do not reduce the set of implementable gates. 
    
    (1) An additional layer can always be turned into the identity operation by measuring all qubits in the $X$ basis $(\alpha=0)$.

    (2) An additional disconnected wire adds an input qubit and applies a single-qubit gate to it, leaving the action on the initial input qubits unaffected.

    (3) As we can see from Eq.~\eqref{eq:unitarymuta} and the subsequent paragraph in Appendix~\ref{append:Generators}, adding a triangle $T_{i,j}$ replaces the gate $\e^{i\alpha_{j,1} X_j/2}$ with $\text{Ising}X\!X_{ij}(-\alpha_{j,1})$, where $\alpha_{j,k}$ denotes the measurement angle at qubit $Q_{j,k}$. For $\alpha_{j,1}=0$, this amounts to removing $\e^{i\alpha_{j,1} X_j/2}$ from $U_j=\e^{i\alpha_{j,3} X_j/2}\e^{i\alpha_{j,2} Z_j/2}\e^{i\alpha_{j,1} X_j/2}\e^{i\alpha_{j,0} Z_j/2}$ acting on qubit $j$ of the input to layer $l$. As a result, wire $j$ of layer $l$ can no longer implement arbitrary single-qubit unitaries. To restore the expressivity, we add layer $l'$, which appends $\e^{i\alpha_{j,0}' Z_j/2}$ to $U_j$ when qubit $Q'_{j,0}$ in layer $l'$ is measured at $\alpha_{j,0}'$ and all other qubits in $X$. Note that it is sufficient to add a single layer $l'$, no matter how many triangles are added to layer $l$.
\end{proof}

\noindent\emph{Remark.} If we make MuTA layers one column deeper---i.\,e., append one qubit to each wire in each MuTA layer---inserting layers when adding triangles becomes unnecessary.

%\begin{proof}
%An additional MuTA$^\prime$ layer implements the identity (up to a global phase) if columns 0, 1 are measured at $\alpha=0$ and columns 2--4 at $\alpha=-\pi/2$. 
%\end{proof}

\begin{property}[2-Colorable]\label{prop:2color}
The graph of any MuTA is 2-colorable.
\end{property}
\begin{proof}
A graph is 2-colorable if and only if it is bipartite. Collecting all even columns into one subset and all odd columns into the other yields a bipartition of any MuTA.
\end{proof}

These properties show that MuTA possesses several useful features of classical NNs which may be missing in other QNN proposals. Specifically, Property~\ref{prop:universality} aligns with Feature I (universality), Property~\ref{prop:tunentanglement} with Feature II (message passing), Property~\ref{prop:bias} with Feature III (bias), Property~\ref{prop:scalability} with Feature IV (scalability), and Property~\ref{prop:monotonexpressivity} with Feature V (monotonicity of expressivity). Using 2-colorable graphs (Property~\ref{prop:2color}) is advantageous thanks to scalable purification protocols for the corresponding states~\cite{robert_bicolorable}. These characteristics position MuTA as a promising universal QNN candidate within the MBQC framework.

Note that the properties listed above should not be taken for granted. For example, generic graph states do not support deterministic computation for arbitrary measurement patterns. In the subset of graphs with guaranteed determinism, there is generally no monotonicity of expressivity, neither with respect to the number of nodes nor edges. Graphs consisting of wires with staggered vertical nearest-neighbor connections carrying an additional node each, which are derived from the 2-dimensional cluster state and are commonly used to prove the universality of MBQC~\cite{raussendorf_review_2012}, also do not ensure determinism. Removing the additional nodes on the vertical connections yields a deterministic and universal ansatz similar to MuTA. In particular, taking two wires of length 7 vertically connected in columns 1 and 3 and measuring the qubits $Q_{0,2}$, $Q_{0,3}$, $Q_{1,0}$, and $Q_{1,6}$ in $X$ yields MuTA layer $(2,0)$~\cite{briegel_graphstates_2004}. However, since only neighboring wires are directly connected in the staggered ansatz, propagating information across its width requires a depth that scales with the number of input qubits. On the other hand, adding more connections to MuTA by allowing the tips of the triangles to be located on different wires within a single layer generally yields graphs without flow.

When it comes to variational circuits, a popular choice is to alternate layers of tunable single-qubit gates with layers of fixed multi-qubit gates such as CNOTs~\cite{mcclean_barren_2018}. However, such ansätze are usually not monotonically expressive, which makes it impossible to identify the required amount of quantum resources by gradually expanding the ansatz. More generally, variational algorithms composed of tunable few-qubit gates do not permit fast information spreading. A resolution based on arbitrary multi-qubit gates leads to an exponential number of variational parameters per fully connected layer~\cite{beer_training_2020,bondarenko2020quantum}, and it can be unclear how to restrict the gates to a sensible number of parameters without jeopardizing essential features such as universality~\cite{torrontegui2019unitary}. On the other hand, Properties~\ref{prop:tunentanglement} and~\ref{prop:scalability} ensure that, for each classical feed-forward NN, we can construct a MuTA with a corresponding connectivity and a comparable number of parameters (see Figure~\ref{fig:classical_to_muta}).
% [Dmytro, I suppose this mapping needs to be explained in a bit more detail. Maybe you could also sketch the figure such that Luis knows what you want him to render? I have moved the open question into the discussion section.]

\begin{figure}[ht]
    \centering
    \subfloat[]{{\includegraphics[width=0.22\textwidth]{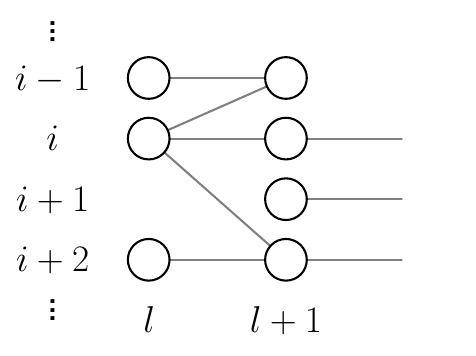} }
    \label{fig:classical_ff}}%
    \subfloat[]{{\includegraphics[width=0.27\textwidth]{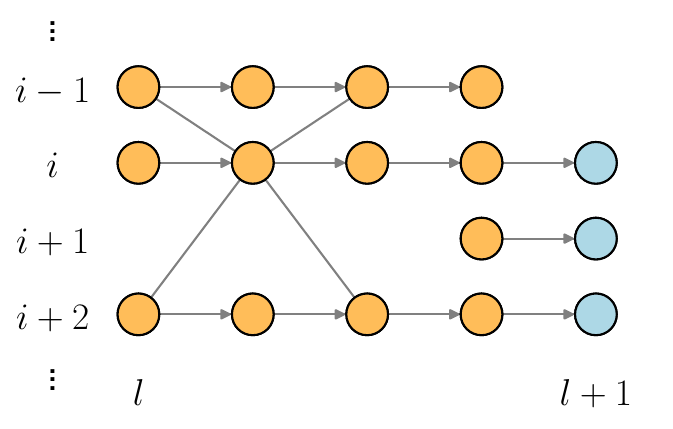} }
    \label{fig:muta_ff}}%
    \caption{A MuTA layer implementing the same connectivity as a classical NN. (a) Layers $l$ and $l+1$ of a feed-forward NN with connectivity between neurons in nodes $i$ and $J = \{i-1, i+2\}$. (b) A MuTA layer $(n,i)$ with connectivity $(i,J)$. 
    Quantizing the same classical NN yields several MuTAs---neurons in subsequent layers can either share a wire or not.
    }
    \label{fig:classical_to_muta}
\end{figure}

\section{Learning to measure \label{sec:Experiments}}

In the preceding section, we have established multiple favorable properties of MuTA, our ansatz for MB-QML. Now, we numerically demonstrate its practical applicability. First, we explore various QML tasks by training the MuTA measurement angles (cf. Fig.~\ref{fig:muta_combined}). In particular, we perform supervised learning on quantum data to learn a universal set of gates and to study the effect of noise. We also classify quantum states according to their metrological usefulness, and learn a quantum instrument that implements teleportation between MuTA wires. Additionally, we introduce an MBQC-specific quantum kernel based on MuTA and use it in a support vector machine (SVM).

\subsection{Universal set of gates}\label{sec:universalgates}

Property~\ref{prop:universality} proves that MuTA can represent a universal set of gates. Here, we demonstrate that such a gate set can easily be learned. Specifically, we show how MuTA layer $(2,0)$, see Fig.~\ref{fig:oneLayerMuta},
learns Haar-random single-qubit gates on the first qubit, as well as the entangling gate $\text{Ising}X\!X(\pi/2)$. For each unitary $U$, the dataset used to learn the measurement angles is 
$\mathcal{D}=\left\{\ket{\psi}_i, U\ket{\psi}_i \right\}_{i=1}^{N}$ and the loss function to minimize is the average infidelity $\ell(\bm{\alpha},\mathcal{D})$ defined in Eq.~\eqref{eq:loss}.

In each run, we use $N=10$ Haar-random states $\ket{\psi}_i$, $N_{\rm tr}=7$ for training and $N_{\rm te}=3$ for testing. In the case of single-qubit unitaries, each run also corresponds to an independent Haar-random choice of $U$. We perform gradient-based optimization using Adam~\cite{kingma2017adam}. Figures~\ref{fig:learningcurve_randomunitaries} and~\ref{fig:learningcurve_isingXX} show the loss over training, averaged over 20 runs with random initializations, for single-qubit unitaries and the $\text{Ising}X\!X$ gate, respectively. In both cases, we observe rapid convergence.

Let us emphasize that, because of the (quantum) no free lunch theorem~\cite{poland_no_2020} and barren plateaus~\cite{mcclean_barren_2018, ragone2023unified}, universality always has to go hand in hand with the ability to implement a bias and to tailor it to specific tasks and prior knowledge, such as symmetries of the data~\cite{schatzki2022theoretical}. To guide the bias engineering, we relate MB-QML ansätze with their infinite-depth expressivity in Appendices~\ref{append:Generators} and~\ref{append:LieAlg}.

\begin{figure}[ht]
    \centering
    \subfloat[]{{\includegraphics[width=0.229\textwidth]{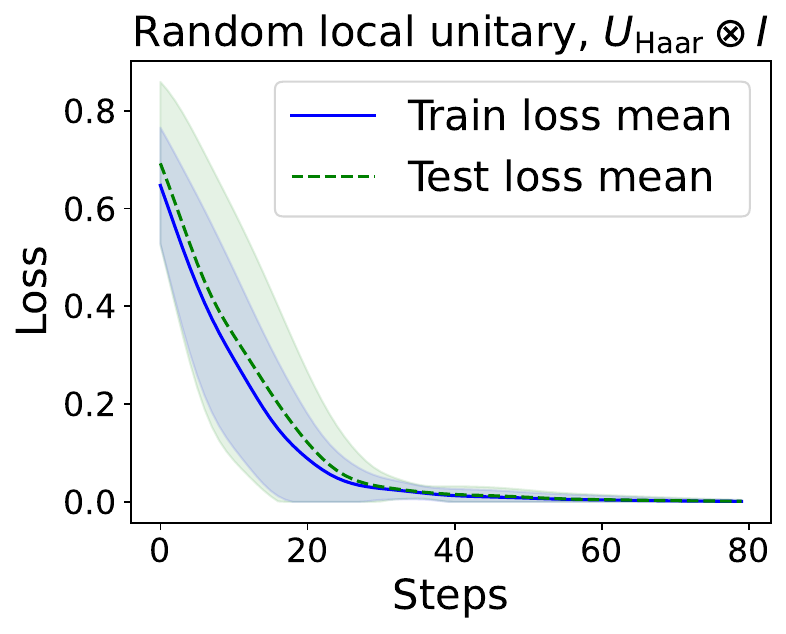} }
    \label{fig:learningcurve_randomunitaries}}%
    \subfloat[]{{\includegraphics[width=0.229\textwidth]{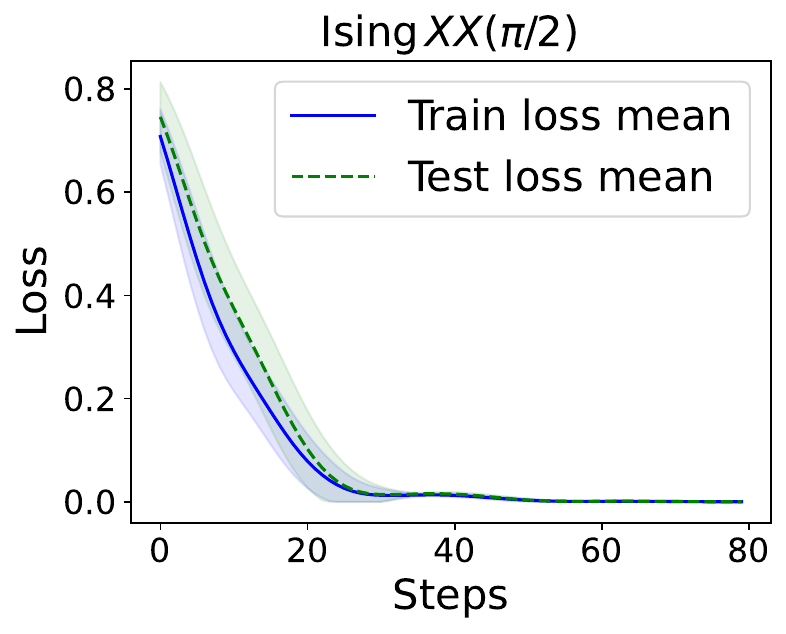} }
    \label{fig:learningcurve_isingXX}}%
    \caption{Learning curves for (a) single-qubit Haar-random unitaries on the first qubit and (b) an $\text{Ising}X\!X\left(\pi/2\right)$ gate, averaged over $20$ runs. The shaded region is the standard deviation.}%
    \label{fig:learningcurves_universal}
\end{figure}

\subsection{Learning with noise}\label{sec:noise}

A criterion for QML to suit current quantum hardware is robustness to noise. We perform the same optimization routine as in Fig.~\ref{fig:learningcurve_isingXX} for two noisy scenarios. First, training the model on noisy data 
$\left\{\ket{\psi}_i,V_iU\ket{\psi}_i\right\}$, where the random unitaries $V_i$ are picked according to a specific noise model, and second, evolving the underlying graph state with a depolarizing channel, which we discuss in Appendix~\ref{append:Depolarizing}.

For the case of noisy data, we consider two types of noise. We approximate the depolarizing channel by applying Brownian circuits $V_i=\prod_{j=0}^{r} e^{i H_{i,j} \Delta t}$. Here, the system briefly evolves for time steps $\Delta t$ under random Hamiltonians $H_{i,j}$, see Ref.~\cite{zhou2019operator} for details. We define the corresponding noise strength as $\frac{\Delta t}{2\pi} \sqrt{2^n r}$. We also study the effect of  bit-flip noise, where $V_i=\bigotimes_jX_j^{\lceil p-f_i\rceil}$ applies a bit flip to qubit $j$ with probability (noise strength) $p$, i.\,e., for $f<p$ with $f\in [0,1)$ uniformly distributed. 

In Figs.~\ref{fig:browniannoise} and \ref{fig:bitflipnoise}, we plot the optimized average fidelity as a function of the noise strength for Brownian and bit-flip noise, respectively. These figures show that increasing the noise strength reduces the fidelity of the trained circuit slowly, as observed in previous QNN proposals~\cite{bondarenko2020quantum, garcia2024effects}. For bit-flip noise, the model can learn the correct parameters as long as the majority of the two-qubit dataset is not flipped---i.\,e., for $p<1-\frac{1}{\sqrt{2}}\approx 0.29$.
\begin{figure}[ht]
    \centering
    \subfloat[]{{\includegraphics[width=0.22\textwidth]{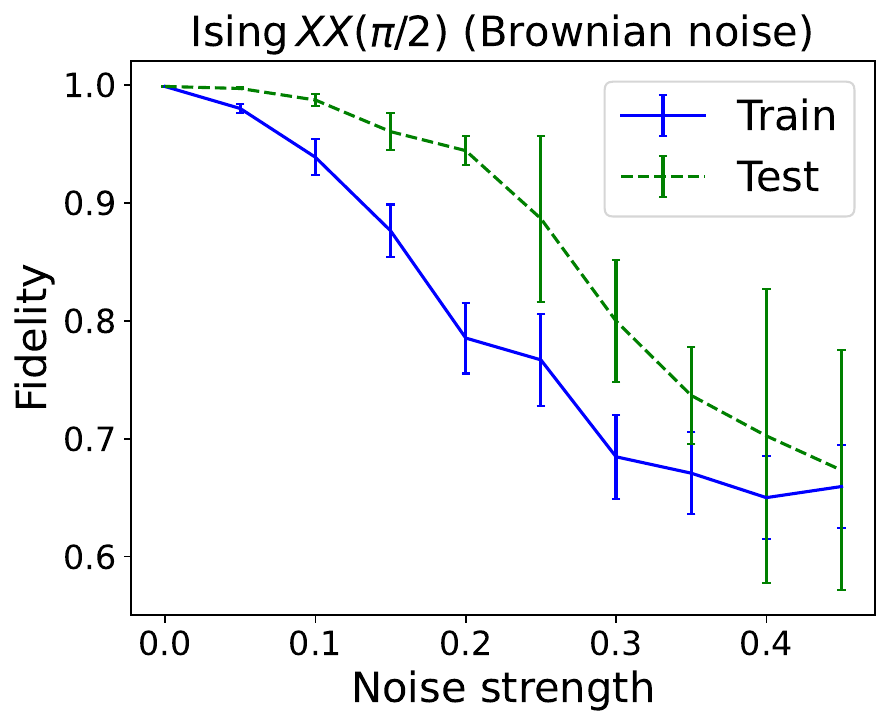} }
    \label{fig:browniannoise}}%
    \quad
    \subfloat[]{{\includegraphics[width=0.22\textwidth]{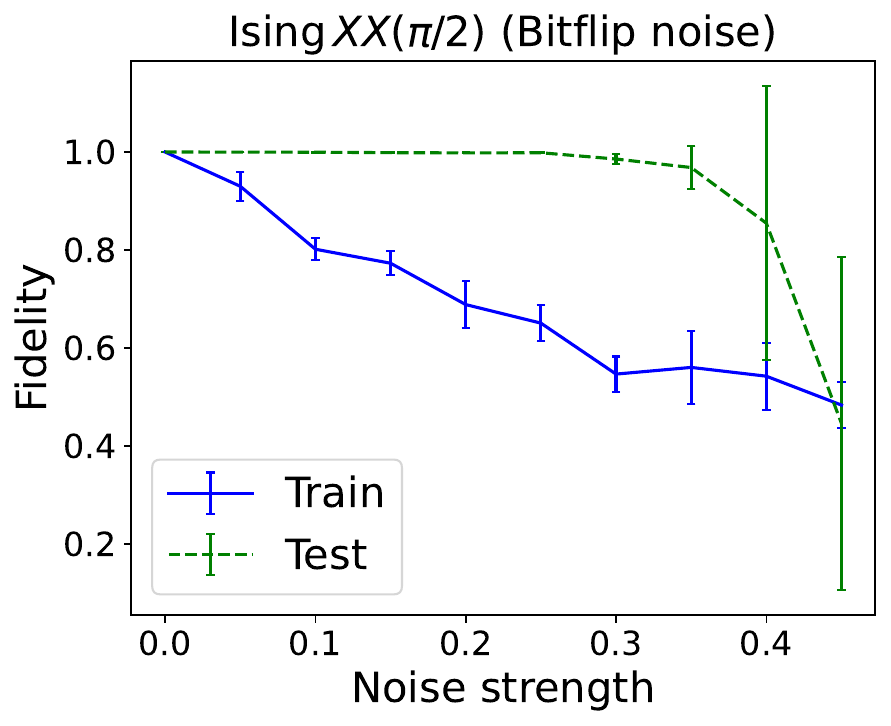} }
    \label{fig:bitflipnoise}}%
    \caption{Stability of MB-QML under noise affecting the training dataset. The testing curves are obtained using noiseless data. Models are trained for $60$ and $200$ steps for Brownian and Bitflip noise, respectively. Each datapoint is an average over $5$ runs, each of which uses a different Haar-random dataset. The dataset size is $N=20$ for Brownian noise and $N=100$ for bit-flip noise. In both cases, the data are split evenly between training and testing.}%
    \label{fig:noisesdataset}
\end{figure}

\subsection{Classifying quantum states}

Classification of complex data constitutes a paradigmatic task for classical~\cite{bishop2006pattern} and quantum~\cite{biamonte_quantum_2017} machine learning. To classify quantum states, we subdivide the task into a positive operator valued measurement (POVM) and subsequent classical post-processing. This ansatz is universal since it includes quantum state tomography followed by arbitrary classical feature extraction.

QML can find applications in quantum sensing and metrology \cite{degen_quantum_2017, giovannetti_advances_2011, caro_out--distribution_2023, li_entanglement_2013, feldmann2021}. Therefore, we train a hybrid quantum-classical model to classify two-qubit states according to their quantum Fisher information (QFI). Quantum metrology exploits entangled probe states $\ket{\psi(\theta)}$ to better estimate an unknown parameter $\theta$. The QFI $F_Q$ asymptotically limits the variance of any estimator $\hat{\theta}$ of $\theta$ via the quantum Cramér-Rao bound: $\operatorname{Var}(\hat{\theta}) \geq [m F_Q(\theta)]^{-1}$, where $m$ is the (sufficiently large) number of experiments. For separable probes, the QFI is tightly bounded by the \emph{standard quantum limit} (SQL), while entangled states allow to surpass the SQL up to the \emph{Heisenberg limit} (HL). Our target classifier $\mathcal{C}$ assigns outcome $0$ to states with a QFI below the SQL and outcome $1$ to states with a QFI above the SQL.

We consider pure states probing $\theta$ through $\ket{\psi(\theta)}=\operatorname{e}^{-i\theta H}\ket{\psi(0)}$ with $H$ a Hermitian operator. The QFI is then given by $F_Q=4\operatorname{Var} H$, where the variance is with respect to $\ket{\psi(0)}$~\cite{toth2013extremal}. For our 2-qubit example, we set $H=h\otimes \mathbf{1} + \mathbf{1} \otimes h$ with $h=\alpha_x X + \alpha_y Y + \alpha_z Z$ and $\alpha_x^2 + \alpha_y^2 + \alpha_z^2 = \frac{1}{4}$, which yields a SQL and HL of $F_Q^{\rm(SQL)}=2$ and $F_Q^{\rm(HL)}=4$, respectively~\cite{feldmann2021}.

The QFI is thus a quadratic function of the probabilities $p^\pm$ to record outcomes $\pm 1$ when measuring $H$ in $\ket{\psi}$. These probabilities can equivalently be extracted by measuring $U\otimes U\ket{\psi}$ in the computational basis if $U$ diagonalizes $h$. To guarantee a sufficient expressivity, we specify our ansatz as follows. The state $\ket{\psi}$ to be classified is input into a MuTA layer $(2)$ without triangles. In each column, we restrict the measurement angles $\boldsymbol{\alpha}$ on the two wires to be equal. We identify $p^\pm$ with the probabilities of outcomes $(0,0)$ and $(1,1)$ for a $(Z,Z)$ measurement of the output state and estimate the QFI by a degree two polynomial $f_{\boldsymbol{\beta}}$ in $p^\pm$ with coefficients $\boldsymbol{\beta}$, $\hat F_{\boldsymbol{\alpha},\boldsymbol{\beta}}(\ket{\psi})=f_{\boldsymbol{\beta}}[p^+_{\boldsymbol{\alpha}}(\ket{\psi}),p^- _{\boldsymbol{\alpha}}(\ket{\psi})]$. Using training data $\{\ket{\psi}_i,y_i\}_{i=1}^N$ with $y_i=0$ ($y_i=1$) for $F_Q(\ket{\psi}_i)\leq 2$ ($F_Q>2$), we optimize the variational parameters $\boldsymbol{\alpha},\boldsymbol{\beta}$ by minimizing the soft-margin loss function
\begin{multline}
    \ell(\boldsymbol{\alpha},\boldsymbol{\beta}) = \frac{1}{N} \sum_{i=1}^N \bigg[ y_i \max \left\{0,-\hat F_{\boldsymbol{\alpha},\boldsymbol{\beta}}(\ket{\psi}_i) + 2 + \epsilon   \right\} + \\
    (1-y_i)\max \left\{0,\hat F_{\boldsymbol{\alpha},\boldsymbol{\beta}}(\ket{\psi}_i) - 2 + \epsilon \right\} \bigg],
\end{multline}
where the parameter $\epsilon$ controls the width of the margin. This loss function is differentiable almost everywhere and has a non-zero gradient when states are misclassified. Setting $\epsilon>0$ avoids the spurious solution $\hat F_{\boldsymbol{\alpha},\boldsymbol{\beta}}=2$.

We sample our set of training states $S = S_1 \cup S_2$ from two two-parameter families,
\begin{align*}
    S_1 &= \left\{\cos{\theta} \ket{00} + e^{i\phi} \sin{\theta}\ket{11} \mid \phi, \theta \sim \mathcal{U}(0, 2\pi)\right\}, \\
    S_2 &= \left\{\cos{\theta} \ket{++} + e^{i\phi} \sin{\theta}\ket{--} \mid \phi, \theta \sim \mathcal{U}(0, 2\pi) \right\},
\end{align*}
where $\mathcal{U}(0, 2\pi)$ is the uniform distribution. These states cover a wide range of QFI for all $h$. Additionally, they can be prepared with a shallow circuit and are spread across a 2-dimensional surface in $(p_+,p_-)$ space for almost every measurement basis.
Setting $h=Z$ and assembling a dataset of $50$ states in $S_1$ and $50$ in $S_2$, we train our model on 80\% of the data using the Adam optimizer and $\epsilon=0.5$. The resulting model obtains a classification accuracy of $0.9725 \pm 0.0042$ on the remaining 20\% of the dataset when ignoring all states that have an estimated QFI between $1.9$ and $2.1$. Figure~\ref{fig:fisherclassifier} shows the classification of Haar-random states, which yields an accuracy of $96\%$ and $99\%$ including or excluding states with $1.9<\hat F<2.1$, respectively. 

\begin{figure}[ht]
    \centering
    \includegraphics[width=0.42\textwidth]{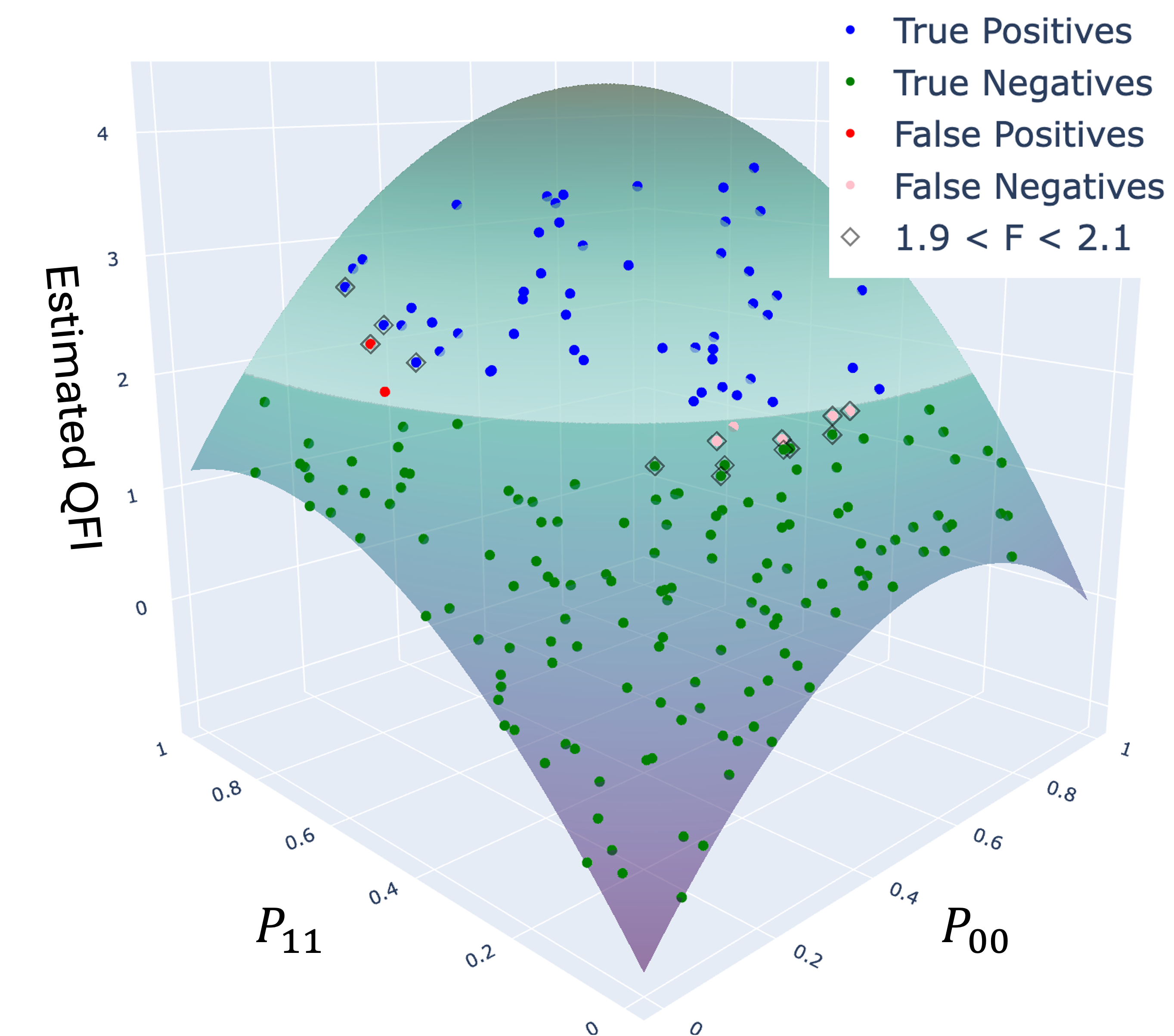}
    \caption{
        Learned surface to classify states of high $({>2})$ and low (${<2}$) QFI. The surface is defined by a degree two polynomial, and the position of each point on the surface is defined by the measurement pattern. The scattered points correspond to Haar-random states.
    }
    \label{fig:fisherclassifier}
\end{figure}

Extracting the QFI for a known operator $h$ does not require any machine learning. However, our example suggests that MB-QML has the capability to classify quantum states even in the absence of an operational definition for the figure of merit. This becomes particularly valuable if we can easily classify a subset of states but lack an efficient method for a larger set of interest. 

\subsection{Quantum instrument for teleportation}

Next, we use MuTA with measurements conditioned on prior measurement outcomes---beyond the adaptation required for determinism---to learn a simple quantum instrument. In general, a quantum instrument $\mathcal{I}=\{(x,E_x)\}_{x\in X}$ is a collection of labels $x\in X$ and corresponding completely positive, trace-non-increasing maps $E_x:\operatorname{End}(\mathcal{H}_1) \rightarrow \operatorname{End}(\mathcal{H}_2)$ such that $\sum_{x\in X} E_x$ is trace-preserving~\cite{gudder2023quantum}.
Quantum instruments describe measurements on quantum states $\rho$ that yield outcome $x$ and post-measurement state $E_x(\rho)/p_x(\rho)$ with probability $p_x(\rho)=\operatorname{Tr}\large[E_x(\rho)\large]$. 
%Equivalently, we can identify them with reductions of projective measurements followed by conditional unitaries on an enlarged Hilbert space \cite{ozawa1984quantum}. 

The instrument we study is the teleportation protocol. This protocol proceeds by creating a maximally entangled two-qubit state $\ket{\Phi}$, measuring the one-qubit state $\ket{\psi}$ to be teleported and half of $\ket{\Phi}$ in an entangled basis, and finally applying a unitary transformation conditioned on the measurement outcome to the unmeasured qubit. In our ansatz, specified in the inset and caption of Fig.~\ref{fig:teleportation-lc}, we represent each of these three steps by an appropriate MuTA layer.

\begin{figure}[ht]
    \centering
    \includegraphics[width=0.4\textwidth]{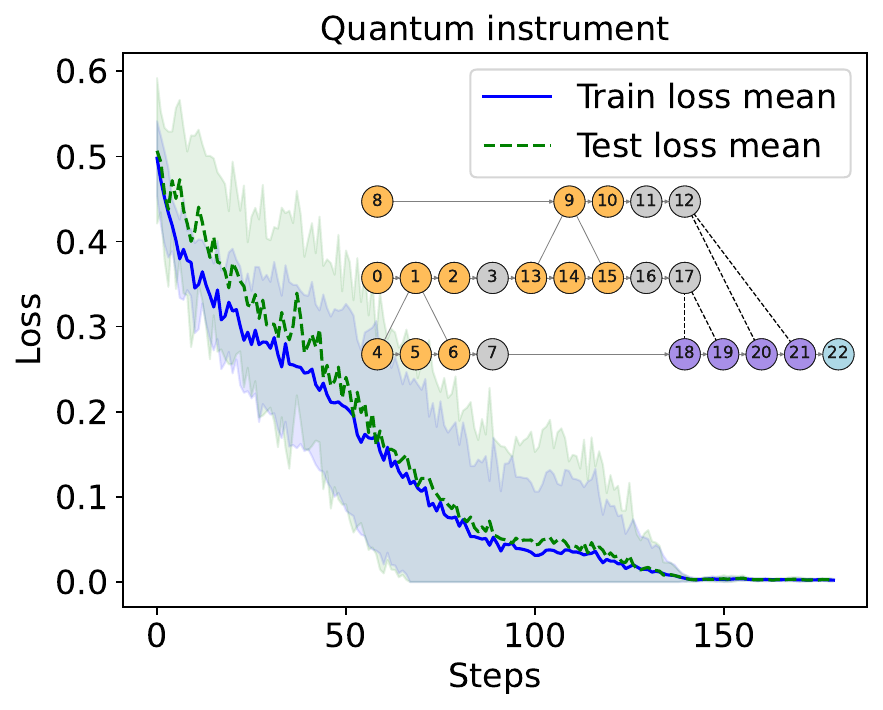}
    \caption{Learning curve for learning a quantum instrument that performs teleportation using MuTA. We average the loss over $10$ runs, each using a different Haar-random dataset with $35$ training states and $15$ test states. The inset shows our ansatz. The node colors yellow, grey, purple, and light blue indicate trainable, fixed, controlled, and output nodes, respectively. Nodes $0$ and $4$ are initialized in $\ket{0}$. The measurements on fixed nodes are $Z$ for nodes 12 and 17 and $X$ otherwise. Controlled nodes are measured in the $XY$ plane at a learnable angle if the controlling measurement outcomes are 1 and in the $X$ basis otherwise.}
    \label{fig:teleportation-lc}
\end{figure}

Figure~\ref{fig:teleportation-lc} shows that the model learns to teleport states from qubit~8 to qubit~22 with unit fidelity. This provides a proof of principle for learning quantum instruments using MuTA and, more generally, for incorporating classical co-processing into MB-QML.

\subsection{Classifying classical data}

Another application and important benchmark for QML is the processing of classical data. There has been extensive work on quantum models for classical classification, and some examples can be found in Refs.~\cite{havlicek_supervised_2019, schuld_quantum_2019, abbas_power_2021}. Here, we create a kernel SVM classifier based on MuTA and use it to classify synthetic datasets.

\begin{figure}[ht]
 \centering
    \includegraphics*[width=0.48\textwidth]{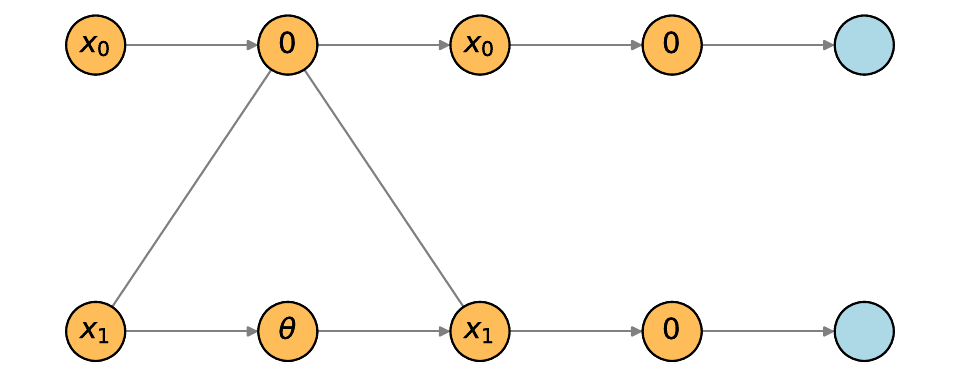}
 \caption{Embedding for kernel SVM classification of 2D datasets. Here, $x_0$ and $x_1$ are the coordinates of the input data, $\theta=\cos(x_0) \cos(x_1)$, and the input state is $\ket{00}$.} 
 \label{fig:kernelSVMMBQC}
\end{figure}

A kernel SVM classifies data points $x$ according to the sign of $\sum_{i=1}^N a_iy_iK(x,x_i) +b$, where $\{(x_i,y_i)\}_{i=1}^N$ is the set of training data, $a_i$ and $b$ are the variational parameters optimized during training, and $K(x,x')$ is the kernel function. In QML, the kernel is typically defined via an embedding (feature map) $x \mapsto \ket{\phi(x)}$ as 
\begin{equation}
  \label{eq:kernel}
  K(x, x') = \abs{\braket{\phi(x)}{\phi(x')}}^2.
\end{equation}
Focusing on 2-dimensional datasets, we adapt an embedding studied in Ref.~\cite{suzuki_analysis_2020} to naturally fit the structure of a single MuTA layer. Figure~\ref{fig:kernelSVMMBQC} shows the measurement angles corresponding to our feature map, explicitly given in Appendix~\ref{append:Kernels}.  

When classifying the different toy datasets in Fig.~\ref{fig:threeclassifiers}, we find that our kernel can be used for data that does not have a strongly non-linear structure, such as \emph{circles} and \emph{blobs}, but not for \emph{moons}. We obtain similar accuracies as Ref.~\cite{suzuki_analysis_2020}. Designing MBQC-specific kernels that can deal with datasets of higher dimension and stronger non-linearity, possibly using kernel alignment \cite{gentinetta2023quantum}, remains an interesting open problem.

\begin{figure}[ht]
  \centering
  \includegraphics[width=0.483\textwidth]{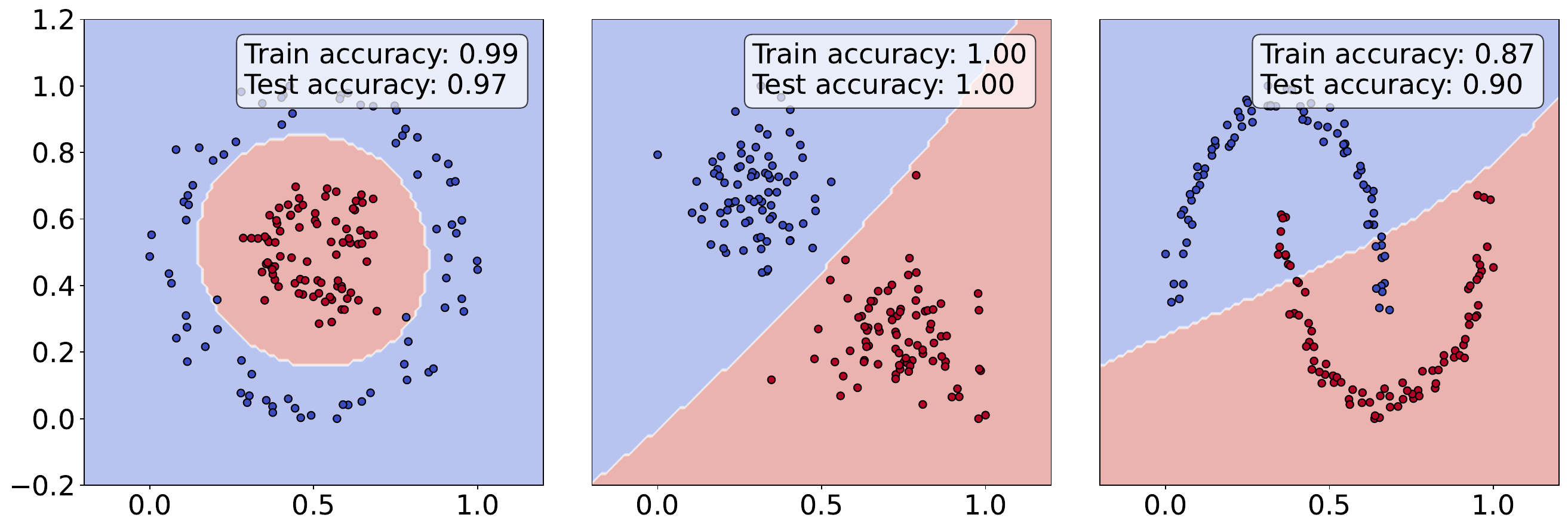}
  \caption{Classification of three different datasets using a SVM with the embedding defined in Fig.~\ref{fig:kernelSVMMBQC} and Eq.~\ref{eq:embedding}. All models were trained on $160$ data points and tested on $40$ data points.}
  \label{fig:threeclassifiers}
 \end{figure}

\section{Hardware-Efficient Ansatz \label{sec:hea}}
Practical implementations of MBQC, e.\,g., with photons, impose hardware restrictions. In particular, the Gottesman-Kitaev-Preskill (GKP) encoding~\cite{gottesman2001encoding} constrains the accessible measurement bases, whereas the dual-rail encoding renders CZ gates probabilistic~\cite{knill2001scheme, shadbolt2012generating}. For MB-QML, this entails either optimizing over a discrete set of measurement angles or using a probabilistic graph state. Here, we focus on photonic GKP computation, where homodyne detection realizes measurements with $\alpha\in\{0,\pi/2\}$, and injecting a $\ket{T}\propto\ket{0} + e^{-i\pi/4}\ket{1}$ magic state supplements $\alpha=\pi/4$ for universality~\cite{knill2004fault, bravyi2012magic}. Below, we propose heuristic algorithms for the discrete optimization over measurement patterns where the permitted angles are $\alpha \in \left\{0, \pi/4, \pi/2 \right\}$. 

Discrete optimization can be NP-hard. Therefore, finding a heuristic for this hardware-efficient ansatz (HEA) is of great practical value. We propose two different approaches. The first one is an epsilon-greedy algorithm that performs a slice-wise search over the measurement angles. We choose slices compatible with the temporal order of the ansatz's flow. In the first round, we randomly initialize all measurement angles and then sequentially replace each slice by its local optimum. If not succeeding, we proceed to optimize two slices at a time. Once the number of slices for simultaneous optimization reaches $L_{\rm{max}}$, we reinitialize and start over (see Appendix~\ref{append:hea} for a pseudo code). Our second approach uses the deep Q-learning algorithm (DQN) from reinforcement learning (RL) \cite{mnih2015human}. Here, the state of the RL agent is defined by an entire measurement pattern and an index $j$ specifying which measurement to modify next. The agent acts by selecting a measurement angle from $\{0, \pi/4, \pi/2\}$ for qubit $j$. Each action yields a new measurement pattern and a reward equal to the average fidelity $1-\ell$, cf. Eq.~\eqref{eq:loss}. We learn the gate $U = (T \otimes \mathbf{1})\cdot \text{Ising}X\!X(-\pi/4)$ with $T=\ket{0}\!\bra{0}+\operatorname{e}^{-i\pi/4}\ket{1}\!\bra{1}$, which can be implemented on a single $(2,0)$-MuTA layer by injecting $\ket{T}$ states in nodes $2$ and $6$ of Fig.~\ref{fig:oneLayerMuta}~\cite{gottesman1999demonstrating}. 

Figures~\ref{fig:learningcurvesGreedyAlgorithm} and~\ref{fig:deepQlearningLC} evaluate HEA training with the epsilon-greedy algorithm and DQN, respectively. We display the minimum infidelity loss, cf. Eq.~\eqref{eq:loss}, against the number of evaluated measurement patterns. This value corresponds to the best measurement pattern found, with respect to the training set, during the execution of either algorithm. Both algorithms generally outperform worst-case random search, making them of practical interest.
For the small learning task under consideration, the epsilon-greedy algorithm shows significantly better convergence than DQN. However, we expect DQN to be advantageous for large input states and structured target measurement patterns. Additionally, the RL approach allows to include actions that variationally depend on the outcomes of mid-circuit measurements.

\begin{figure}[ht]
 \centering
 \includegraphics[width=0.4\textwidth]{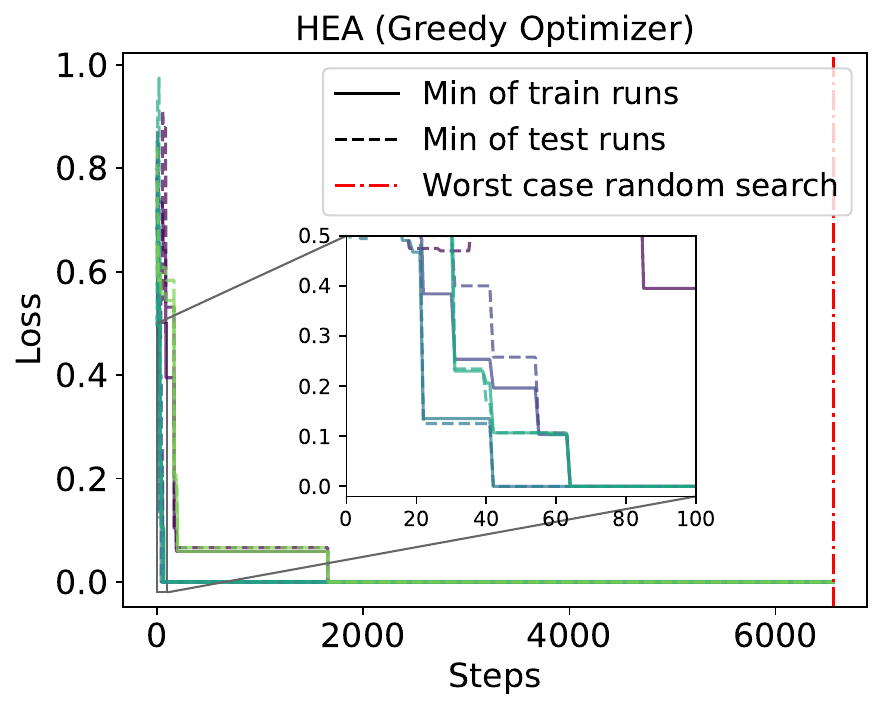}
 \caption{Learning curves for five runs of the greedy algorithm compared to random search. All examples set $\epsilon=0$ and $L_{\rm{max}}=4$. Two out of five runs failed to progress until the parameters were reinitialized. The training and testing datasets are of sizes $N_{\rm train}=7$ and $N_{\rm test}=3$.}
 \label{fig:learningcurvesGreedyAlgorithm}
\end{figure}

\begin{figure}[ht]
 \centering
 \includegraphics[width=0.4\textwidth]{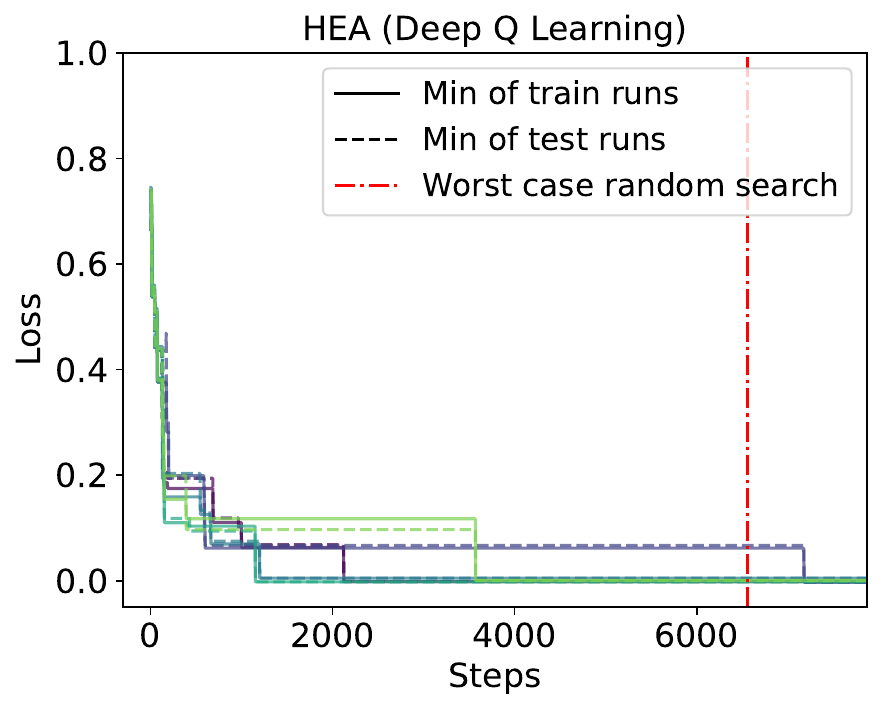}
 \caption{Learning curves for five runs of the DQN algorithm compared to random search. The training and testing datasets are of sizes $N_{\rm train}=21$ and $N_{\rm test}=9$.}
 \label{fig:deepQlearningLC}
\end{figure}

\section{Discussion}\label{sec:discussion}

In this work, we studied how to use MBQC for QML. In particular, we have proposed a framework for MB-QML and introduced a QNN, MuTA, which stands out by simultaneously satisfying numerous desirable properties including universality, tunable entanglement, monotonic expressivity, and versatile bias engineering. We provided numerical evidence for the practical value of MuTA by learning a universal gate set, a quantum state classifier, and a quantum instrument. We also showed how to use MuTA to generate kernel matrices for SVM classification of classical data. Finally, we proposed training algorithms for MuTA that account for the hardware constraints of photonic GKP qubits. 

MuTA facilitates the theoretical analysis of QML in the MBQC framework. As a flexible general-purpose QNN designed for MBQC, it is a promising candidate for QML on near-term MBQC devices. We expect that MuTA will help leverage advantages of MBQC, such as reduced time complexity and compatibility with classical side-processing, for QML.

Several questions arise from this work. Measurements with feedback allow to implement specific maps with fewer resources \cite{raussendorf_long-range_2005, broadbent2009parallelizing, malz_preparation_2023, rossi_using_2021}. MBQC platforms are particularly suited for such algorithms because of the inherent mid-circuit measurements and classical co-processing. However, exploiting these benefits requires a more elaborate model that includes classical and quantum processing. More generally, it would be interesting to explore conditions on the data under which MB-QML is likely to outperform its gate-based counterparts.

Another promising research direction is to apply MB-QML to more complex datasets, including sequential quantum data. This could involve designing recurrent MB-QML models akin to those in Ref.~\cite{bondarenko_learning_2023} or developing attention models based on graph states. Again, mid-circuit measurements with feedback would be key, as they enable the use of classical memory, which can be crucial for resource efficiency.

We focused on deterministic MBQC because channels can be lifted to unitary operations on larger Hilbert spaces (Steinspring dilation~\cite{stinespring1955positive}). However, MBQC also allows to implement channels with fewer quantum resources by leveraging the inherent randomness of measurement outcomes~\cite{majumder2024variational}. It would be interesting to characterize the expressivity and to explore the performance of MuTA equipped with a trainable probabilistic adaptation of measurement angles.

Finally, as observed in Section~\ref{sec:MUTA}, the geometry of a classical feed-forward NN can be mapped onto a corresponding MuTA. This raises the question of whether the resulting MuTA is more expressive than its direct classical counterpart.

\begin{acknowledgments}
The authors would like to thank Alán Aspuru-Guzik, Olivia Di Matteo, Paul Herringer, Joe Salfi, and Ilan Tzitrin for helpful discussions. L.M.C. acknowledges support from the Novo Nordisk Foundation, Grant number NNF22SA0081175, NNF Quantum Computing Programme. P.\,F. acknowledges financial support from the Canada First Research Excellence Fund through the Quantum Materials and Future Technologies Program, and from the Natural Sciences and Engineering Research Council of Canada (NSERC) Quantum Alliance Consortium ``CanQuEST''. D.\,B. thanks support from the Canada First Research Excellence Fund. 
\end{acknowledgments}

\textbf{Data availability:} All simulations are written in Python and can be found at \href{https://github.com/mentpy/mentpy}{https://github.com/mentpy/mentpy}.

\bibliographystyle{apsrev4-2}
\bibliography{apssamp}% Produces the bibliography via BibTeX.

\clearpage
\newpage

\appendix 

\section{Monotonic expressivity of classical neurons}\label{App:ExpressMononton}
When designing a model, it is very convenient if increasing its complexity leads to increased capabilities. For NNs this means that additional neurons or connections should enlarge the variational class. 
%However, there are two caveats that need to be considered when formally defining this property. First, the range and domain of a neuron are often different.
Such monotonicity of expressivity is especially useful if it allows to pre-train parts of a larger model. Examples include training networks layer by layer or amending a trained model with a low-depth post- or pre-processing network. This motivates the additional requirement that the variational parameters of the smaller network, implementing $\mathcal{N}$, can remain unchanged for the enlarged model to reproduce $\mathcal{N}$.

To make this appendix self-contained, let us define classical NNs.
\begin{definition}[Classical artificial neuron]
    Let $\sigma: \mathbb{R} \rightarrow \mathbb{R}$ be a non-linear and non-constant activation function. A map $\mathfrak{n}_{\mathbf{w},b}: \mathbb{R}^n \rightarrow \mathbb{R}$ with variational parameters $\mathbf{w}\in\mathbb{R}^n$, $b \in \mathbb{R}$ such that $\mathfrak{n}_{\mathbf{w}, b}(\mathbf{x}) = \sigma\left(\mathbf{w}\cdot \mathbf{x} + b \right)$ is called a classical artificial neuron.
\end{definition}
\begin{definition}[Classical neural network]
    A neural network is a composition of neurons. If the output of neuron $i$ serves as an input to neuron $j$ we say that $i$ and $j$ are connected.
\end{definition}
\begin{observation}[Monotonicity of connections] Connections between neurons can always be switched off. Indeed, if the output of neuron $i$ becomes input $x_\iota$ of neuron $j$, setting $w_\iota=0$ in $\mathfrak{n}^{(j)}_{\mathbf{w},b}(\mathbf{x}) = \sigma\left(\sum_{k} w_{k}  x_k + b \right)$ disconnects $j$ from $i$. 
\end{observation}
Thus, for classical NNs it is sufficient to concentrate on adding neurons. Assuming that the input data lies within the range of $\sigma$, $x\in\sigma(\mathbb{R})$, we adopt the following
%that are connected to all of the previous outputs.
\begin{definition}[Monotonicity of expressivity]\label{def:monotonicity}
    We say that a NN is monotonically expressive if the activation function $\sigma$ of its neurons fulfills 
    \begin{equation}
        \exists w,b\in\mathbb{R}:\sigma(wx+b) = x\quad\forall x\in \sigma(\mathbb{R})
    \end{equation}
\end{definition}
%\begin{definition}[Monotonicity of expressivity]
%    Consider any NN composed of a given type of neurons with an output $\mathcal{O}$. Let us add a new layer $\tilde{\mathcal{O}}$ with variational parameters $\{ v_i \}$ such that
%    \begin{itemize}
%        \item $\left|\tilde{\mathcal{O}}\right| = |\mathcal{O}|$,
%        \item The new layer is fully-connected, i.e. every neuron in $\tilde{\mathcal{O}}$ is connected to and is next with respect to every neuron in $\mathcal{O}$.
%    \end{itemize}
%    We say that the variational class is monotonically expressive if there exists a choice of $\{ v_i \}$ such that the new network coincides with the old network, i.e.
%    \begin{equation}
%        \exists \{v_i \}: \forall x \in Dom(\mathrm{NN}) \quad \tilde{\mathcal{O}}\vert_{\{v_i\} } \circ \mathrm{NN} (x) = \mathrm{NN}(x). 
%    \end{equation}
%\end{definition}
This definition ensures that the expressivity does not fall as the depth of a network increases. For universal networks such as MuTA or classical feed-forward NNs~\cite{FFNN_Univ_deep_min(width)_arbitrary,FFNN_Univ_deep_min(width)0,FFNN_Univ_deep_distributions} the expressivity actually grows---albeit not always monotonically---with the number of fully connected layers until universality is reached.

Do all kinds of classical artificial neurons give rise to monotonically expressive networks? It turns out that the answer is no: only ReLU neurons and their close relatives do. We speculate that the monotonicity of expressivity is one of the main reasons why ReLU networks are so successful~\cite{ReLU_advantage}.
%\begin{definition}[ReLU-like neurons]
%    Consider a function $\sigma(\cdot): Dom(\sigma) \rightarrow R(\sigma)$. We call a neuron $\mathfrak{n}\vert_{\{w_i \}, b}\left(\{ x_i \}_{i=1}^n\right) = \sigma\left(\sum_{i} w_{ij} \cdot x_i + b \right)^{\times m}$ ReLU-like if
%    \begin{itemize}
%        \item $R(\sigma) \subsetneq Dom(\sigma)$,
%        \item $\sigma(\cdot)$ is affine on $R(\sigma)$.
%    \end{itemize}
%\end{definition}
%Condition $R(\sigma) \subsetneq Dom(\sigma)$ ensures the non-linearity of a ReLU-like neuron. Affine neurons are rarely used as deep networks are just as expressive as shallow ones. This condition can be relaxed when discussing the monotonicity of expressivity.
\begin{definition}[ReLU-like neurons]
    We call a neuron ReLU-like if $\exists \Omega\subset\mathbb{R}$ such that $\sigma\vert_\Omega$ is a non-constant affine function,
    \begin{equation}
        \sigma\vert_\Omega(x) = vx+c,\quad v,c\in\mathbb{R},\;v\neq 0,
    \end{equation}
    and $\sigma(\Omega)=\sigma(\mathbb{R})$.
\end{definition}
\begin{theorem}[Monotonicity of expressivity for classical NNs]
    Only ReLU-like neurons are monotonically expressive.
\end{theorem}
\begin{proof}
    The main observation is that the inverse of an affine function $a(x) = v x + c, \ v \neq 0$ is also affine,
    \begin{equation}
        a^{-1}(x) = w  x + b, \quad w = \frac{1}{v},\ b= -\frac{c}{v}.
    \end{equation}
    Indeed, $a\left[ a^{-1}(x)\right] = v \left( \frac{x}{v}- \frac{c}{v} \right) + c = x$.

    Let us fist prove that ReLU-like neurons are monotonically expressive. Since $\sigma(\Omega)=\sigma(\mathbb{R})$ and $\sigma\vert_\Omega$ is non-constant affine, there exists an affine preimage $f:\sigma(\mathbb{R})\rightarrow \mathbb{R}$ of $\sigma$; this statement is equivalent to monotonic expressivity as defined in~\ref{def:monotonicity}.

    On the other hand, monotonic expressivity requires an affine $f:\sigma(\mathbb{R})\rightarrow \mathbb{R}$ such that $\sigma[f(x)]=x$ $\forall x\in \sigma(\mathbb{R})$. Since $\sigma$ is not a constant function, $f$ cannot be constant, either. Setting $\Omega=f[\sigma(\mathbb{R})]$, we can thus conclude that $\sigma\vert_{\Omega}$ is non-constant affine and $\sigma(\Omega)=\sigma(\mathbb{R})$, such that the neuron is ReLU-like.
\end{proof}

\section{From MBQC to unitary circuits}\label{append:Generators}
Let us consider graph states $|G\rangle$ over open graphs $G$ with flow~\cite{danos_determinism_2006} that have an equal number of input and output qubits, $|I|=|O|=n$. For such graph states, we derive a simple method for translating from MBQC to the circuit model. Note that the input and output sets may overlap, $I\cap O\neq \emptyset$. For convenience, we call qubits $i\in I \setminus (I\cap O)$ and $o\in O\setminus (I\cap O)$ \emph{proper} input and output qubits, respectively.

We begin our construction by observing that:
\begin{enumerate}
	\item Repeatedly applying the flow function $f$ to a proper input qubit $i$ traces a path from $i$ to a proper output qubit. We call such a path an $f$-path.\\
	\emph{Proof: $f$ traces a path because $f(j)$ is a neighbor of $j$ $\forall j$; this path is finite because $f(j)>j$ and the number of qubits is finite; the final qubit of the path must be an output qubit because otherwise we could prolong the path by applying $f$ once more; this output qubit is proper because input qubits are not in the range of $f$.}
	\item Different $f$-paths do not intersect.\\
	\emph{Proof: if $f(i)=f(j)$ for $i\neq j$ then $j$ is a neighbor of $f(i)$ such that $j>i$, and $i$ is a neighbor of $f(j)$ such that $i>j$, which is a contradiction.}
	\item Every qubit $j\notin I\cap O$ lies on an $f$-path. Qubits $j\in I\cap O$ do not lie on $f$-paths.\\
	\emph{Proof: applying $f$ to $n-|I\cap O|$ proper input qubits defines paths ending in $n-|I\cap O|$ different proper output qubits; if $j\notin I\cap O$ does not lie on any of these paths, repeatedly applying $f$ must trace a path from $j$ to another proper output qubit; however, no other proper output qubits exist; $j\in I\cap O$ cannot lie on an $f$-path because it is neither in the domain nor in the range of $f$.}
	\item All $f$-paths are induced paths; i.\,e., within a single $f$-path, $G$ has no edges between non-consecutive qubits.\\
	\emph{Proof: let $j<k$ be non-consecutive qubits in an $f$-path; then there exists $k'>j$ such that $k=f(k')$; if $j$ is a neighbor of $k$ then $j>k'$, which is a contradiction.}
\end{enumerate}
Hence, any $|G\rangle$ consists of $n-|I\cap O|$ interconnected 1-dimensional cluster states connected to $|I\cap O|$ input-output qubits.

To translate from MBQC to the circuit model, we begin by picking a one-by-one measurement order compatible with the flow. The first qubit in the measurement order will be a proper input qubit, $q_1\in I\setminus (I\cap O)$. \emph{(Proof: if $q_1\notin I$ there exists $j$ such that $f(j)=q_1$ and $j<q_1$, which contradicts that $q_1$ is the first qubit; $q_1\notin O$ because output qubits are not measured.)} Every qubit $q\in V\setminus O$ has an edge to its successor $f(q)$. If $q_1$ has additional neighbors, they must lie in $I$; in this case, we call $q_1$ a junction qubit. \emph{(Proof: let $j$ be a neighbor of $q_1$ with $j\neq f(q_1)$, $j\notin I$; then there exist $j'\neq q_1$ such that $j=f(j')$ and $j'<q_1$, which is a contradiction.)} The edges between $q_1$ and other input qubits can be absorbed into the input state by applying the corresponding CZ gates. Measuring $M_\alpha^{q_1}$ then applies $U_1=H\e^{i\alpha Z/2}$ to the updated input state at $q_1$ and teleports it to $f(q_1)$~\cite{raussendorf_review_2012}. We have thus translated the first measurement into a unitary gate.

The obtained circuit defines an input state to a reduced MBQC from which qubit $q_1$ has been removed, $V'=V\setminus \{q_1\}$. The input and output sets of this reduced MBQC are $I'=\left(I\setminus \{q_1\} \right)\cup \{f(q_1)\}$ and $O'=O$, respectively, such that $|I'|=|O'|=n$. Restricting $f$ to $(V\setminus O)\setminus \{q_1\}=V'\setminus O'$ yields a flow function $f'$. The initially picked measurement order, from the second qubit $q_2$ onward, remains compatible with $f'$. Hence, we can proceed by translating the measurement of $q_2$ and, by induction, all subsequent measurements in exactly the same way as we have translated the first measurement. Once the MBQC has been reduced to $O$, applying a CZ gate for each edge between output qubits finalizes the translation.

Let us summarize the translation algorithm. For clarity, we introduce a bijective map from MBQC to circuit input qubits, $\tau:I\to \{0,\ldots,n-1\}$. Then measuring qubit $q_j$ on the $f$-path from $i\in I\setminus (I\cap O)$ translates to $U_j=H_{\tau(i)}\e^{i\alpha_j Z_{\tau(i)}/2}$, where $\alpha_j$ is the measurement angle. The gates $U_j$ act in the measurement order of MBQC. If $q_j$ is a junction qubit---a qubit that retains more than one edge until being measured---we have to insert CZ gates before $U_j$. For a retained edge between $q_j$ and a qubit $q$ on the $f$-path from $i'$ or $q=i'$, $i'\in I$, the specific gate is $\operatorname{CZ}_{\tau(i),\tau(i')}$. Edges between output qubits $o_1,o_2\in O$ translate into $\operatorname{CZ}_{\tau(i_1),\tau(i_2)}$ at the end of the circuit, where $o_j$ is a qubit on the $f$-path from $i_j$ or $o_j=i_j$, $i_j\in I$.

This translation algorithm enables us to express MuTA in terms of the unitary circuit model. A (connected) MuTA layer $(2,0)$, see Fig.~\ref{fig:oneLayerMuta} in the main text, has two junction qubits, $1$ and $5$. Setting $\tau(0)=a$, $\tau(5)=b$, where we use alphabetic labels for better readability, we obtain the unitary circuit
\begin{align}\label{eq:unitarymuta20}
\begin{split}
	U_{(2,0)}={}& H_b\e^{i\alpha_8 Z_b/2}H_b\e^{i\alpha_7 Z_b/2}H_a\e^{i\alpha_3 Z_a/2}H_a\e^{i\alpha_2 Z_a/2}\\
	&H_a\e^{i\alpha_1 Z_a/2}\operatorname{CZ}_{ab}H_b\e^{i\alpha_6 Z_b/2}H_b\e^{i\alpha_5 Z_b/2}\operatorname{CZ}_{ab}\\
        &H_a\e^{i\alpha_0 Z_a/2}\\
	={}&\e^{i\alpha_8 X_b/2}\e^{i\alpha_7 Z_b/2}\e^{i\alpha_3 X_a/2}\e^{i\alpha_2 Z_a/2}\e^{i\alpha_1 X_a/2}\\&H_a\operatorname{CZ}_{ab}\e^{i\alpha_6 X_b/2}\operatorname{CZ}_{ab}H_a\e^{i\alpha_5 Z_b/2}\e^{i\alpha_0 Z_a/2}\\
	={}&\e^{i\alpha_8 X_b/2}\e^{i\alpha_7 Z_b/2}\e^{i\alpha_3 X_a/2}\e^{i\alpha_2 Z_a/2}\e^{i\alpha_1 X_a/2}\\&C\!N\!OT_{ba}\e^{i\alpha_6 X_b/2}C\!N\!OT_{ba}\e^{i\alpha_5 Z_b/2}\e^{i\alpha_0 Z_a/2}\\
	={}&\e^{i\alpha_8 X_b/2}\e^{i\alpha_7 Z_b/2}\e^{i\alpha_3 X_a/2}\e^{i\alpha_2 Z_a/2}\e^{i\alpha_1 X_a/2}\\&\text{Ising}X\!X_{ab}(-\alpha_6)\e^{i\alpha_5 Z_b/2}\e^{i\alpha_0 Z_a/2}.\\
\end{split}
\end{align}

Correspondingly, a general fully connected MuTA layer $(n,i)$ implements
\begin{align}\label{eq:unitarymuta}
\begin{split}
	U_{(n,i)}={}&\left(\prod_{k=0}^{n-1} H_k\e^{i\alpha_{k3}Z_k/2}H_k\e^{i\alpha_{k2}Z_k/2}\right)H_i\e^{i\alpha_{i1}Z_i/2}\\
	&\left(\prod_{\substack{j=0\\j\neq i}}^{n-1}\operatorname{CZ}_{ij}H_j\e^{i\alpha_{j1}Z_j/2}H_j\e^{i\alpha_{j0}Z_j/2}\operatorname{CZ}_{ij}\right)\\&H_i\e^{i\alpha_{i0}Z_i/2}\\
	={}&\left(\prod_{k=0}^{n-1}\e^{i\alpha_{k3}X_k/2}\e^{i\alpha_{k2}Z_k/2}\right)\e^{i\alpha_{i1}X_i/2}\\
	&\left(\prod_{\substack{j=0\\j\neq i}}^{n-1}\text{Ising}X\!X_{ij}(-\alpha_{j1})\e^{i\alpha_{j0}Z_j/2}\right)\e^{i\alpha_{i0}Z_i/2}.
\end{split}
\end{align}
Here, measurement angles $\alpha_{ij}$ refer to qubits $Q_{i,j}$ in row $i$ and column $j$, and the circuit input qubits are enumerated according to the rows of MuTA, $\tau(Q_{i,0})=i$. For an ansatz that is not fully connected, $\text{Ising}X\!X_{ij}(-\alpha_{j1})$ in Eq.~\eqref{eq:unitarymuta} has to be replaced with $\e^{i\alpha_{j1}X_j/2}$ for each row $j$ that is not connected to row $i$ by a triangle $T_{i,j}$.

\section{Expressivity of QML ansätze \label{append:LieAlg}}

The expressivity of parameterized quantum circuits assembled from elementary gates $\exp{-i\theta H_j}$, $H_j\in G$ is determined by the set of generators $G$. Specifically, in the limit of infinite depth, such a circuit implements $\exp{\mathfrak{g}}$, 
% [In the paper you cite, they only write that all circuit operations belong to $\exp{\mathfrak{g}}$. Are you sure that the converse is true?],
where the dynamical Lie algebra of the circuit, $\mathfrak{g}=\operatorname{span}_\mathbb{R}\langle iG\rangle_{[\cdot,\cdot]}$, is the vector space spanned by the Lie closure of $iG$.
% \del{A common way of quantifying the expressivity of a QML model is via its Lie algebra $\mathfrak{g} \subseteq \mathcal{P}_n$. This algebra is generated by a set of generators $H$ which can be obtained as outlined in Appendix \ref{append:Generators}. The set of gates the model can implement is $$\mathbb{G}=\left\{ \prod_j \exp(i\theta_j P_j) \mid \theta_j \in \mathbb{R}, P_j \in \langle H \rangle_{[\cdot, \cdot]} \right\}, $$ where $\langle H \rangle_{[\cdot, \cdot]}$ denotes the Lie closure of $H$. We can then calculate the Lie algebra $\mathfrak{g}$ by taking the closure of $H$ under the Lie bracket $[ \cdot, \cdot]$.} 
This quantity not only characterizes the expressivity of the variational circuit but also allows to predict barren plateaus~\cite{ragone2023unified}. In particular, if $\mathfrak{g}$ is simple, the variance of infidelity loss is bounded by
\begin{align}
    \label{eq:varianceLieAlg}
    \mathrm{Var} [\ell] \leq \frac{2^n}{\dim(\mathfrak{g})}.
\end{align}

For MB-QML, we can determine $G$ as outlined in Appendix~\ref{append:Generators}.
We find that the Lie algebra of the (fully connected) MuTA ansatz with $n$ rows, proposed in Section \ref{sec:MUTA}, is the simple Lie algebra $\mathfrak{g}=\mathfrak{su}( 2^n )$ with $\dim(\mathfrak{g})=2^{2n}-1$. This is a direct consequence of the universality of MuTA. Equation~\eqref{eq:varianceLieAlg} emphasizes the importance of bias in QML. Luckily, MuTA offers versatile bias engineering as summarized in Property~\ref{prop:bias}.

\section{Noisy resource state \label{append:Depolarizing}}
QML protocols can suffer not only from noisy data but also from noise affecting the models themselves. Here, we study the action of the depolarizing channel on the resource state of an MB-QML circuit. For a single qubit, the channel with noise strength $p$ is given by
\begin{equation*}
    \mathcal{N}_{\rm{dep}}(\rho) = (1-p)\rho + \frac{p}{3} X \rho X + \frac{p}{3} Y \rho Y + \frac{p}{3} Z \rho Z.
\end{equation*}
We apply this channel to each MuTA qubit. In the presence of such noise, standard measurement-angle adaptation does, in general, no longer compensate ${({-1})}$-outcomes. Therefore, we take into account all possible combinations of measurement outcomes to obtain the mixed output state $\rho$ of the MB-QML model. We simulate this by propagating the angle adaptation operations using the deferred measurement principle. In the cost function, the fidelity $F=|\!\bra{\psi}\ket{\phi}\!|^2$ is replaced by $\bra{\phi}\!\rho\!\ket{\phi}$, where $\ket{\phi}$ is the desired output state. 

\begin{figure}[ht]
    \centering
    \includegraphics[width=0.4\textwidth]{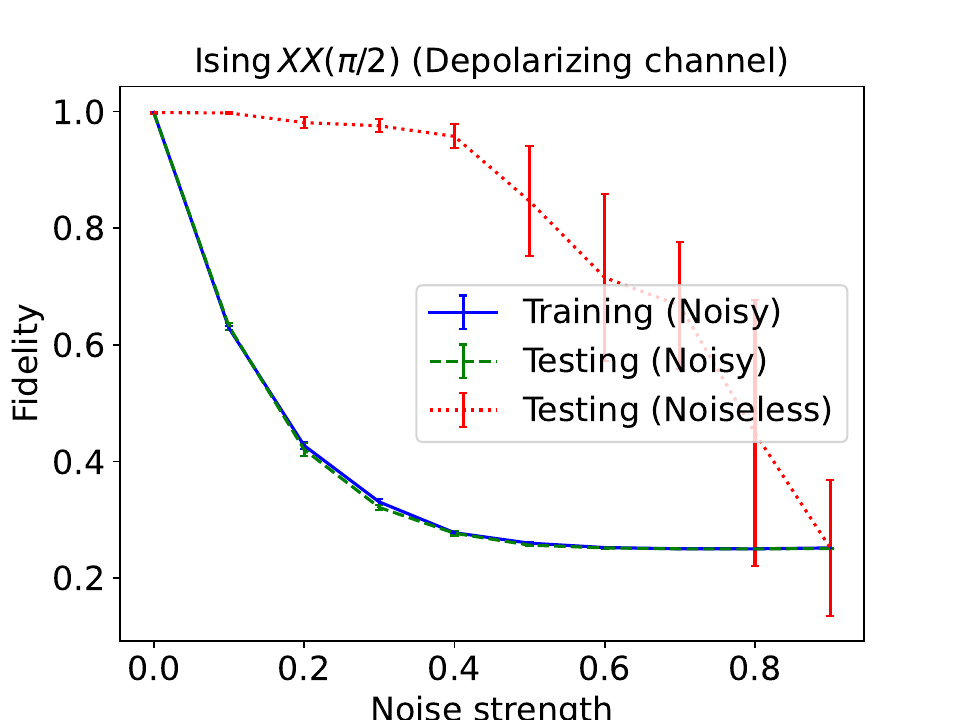}
    \caption{Effect of a depolarizing channel affecting the MuTA qubits in Fig.~\ref{fig:oneLayerMuta} on learning $\text{Ising}X\!X(\pi/2)$. We average over $5$ runs. Each run uses a different dataset of Haar-random states, $10$ for training and $3$ for testing.}
    \label{fig:depolarizingNoise}
\end{figure}

Figure \ref{fig:depolarizingNoise} shows training and testing  outcomes as a function of $p$ for learning the $\text{Ising}X\!X(\pi/2)$ gate on a distorted MuTA, cf. Sections~\ref{sec:universalgates} and~\ref{sec:noise}. Testing on the ideal resource state reveals that the model learns the measurement angles appropriate for noiseless operation up to a quite large value of $p\sim 0.4$. This robustness presents an interesting use case: one could use a less expensive, noisier device for initial training and a more expensive, lower-noise device for fine-tuning and inference.

\section{Quantum kernels with MBQC \label{append:Kernels}}
To design a quantum kernel in MBQC, we choose a graph state and a map from data points $x$ to measurement patterns. Figure~\ref{fig:kernelSVMMBQC} defines the feature map 
\begin{multline}
    \label{eq:embedding}
    \ket{\phi(x)} =  R^z_{1} (-x_{1}) R^z_{0} (-x_{0}) e^{i\cos(x_{0})\cos(x_{1})X_0 X_1 /2}\\
    R^z_{1} (-x_{1}) R^z_{0} (-x_{0}) \ket{00}
\end{multline}
with $R^z_j(\alpha)= \exp(-i\alpha Z_j/2)$. 
This embedding is structurally similar to those studied in Ref.~\cite{suzuki_analysis_2020}.
The kernel function $K(x,x')=\abs{\braket{\phi(x)}{\phi(x')}}^2$ can be evaluated via the SWAP test using two graph states to prepare both $\ket{\phi(x)}$ and $\ket{\phi(x')}$. \\

\section{Training the hardware-efficient ansatz}\label{append:hea}
In Section \ref{sec:hea}, we used heuristic algorithms to optimize the measurement pattern over a discrete set of angles, specifically those accessible in the GKP encoding with $\ket{T}$ magic states. Algorithm \ref{prog:greedyOptimization} presents our epsilon-greedy temporal-slice-wise optimization technique.

\begin{algorithm}[H]
\caption{Slice-wise epsilon-greedy search algorithm for the discrete optimization of measurement patterns.}
\label{prog:greedyOptimization}
\begin{algorithmic}[1]
\Function{GreedyOpt}{$\ell, \Lambda, \varepsilon, n_{\text{reset}}, L_{\max}, \delta$}
    \For{$0:n_{\text{reset}}-1$}
        \State $\theta = \text{randomAngles(\ )}$ 
        \For{$m = 1$ to $L_{\max}$}
            \State $\theta = \text{LayerOpt}(\ell, \theta, \Lambda, m, \varepsilon)$
            \If{$\ell(\theta) < \delta$}
                \State \textbf{return} True, $\theta$                    
            \EndIf
        \EndFor
    \EndFor
    \State \textbf{return} False, $\theta$
\EndFunction
\Function{LayerOpt}{$\ell, \theta, \Lambda, n, \varepsilon$}
    \State $\theta_{\text{best}} = \theta$
    \For{$i = 0$ to $\text{len}(\Lambda) - n$}
        \State $\Lambda' = \text{concatenate}(\Lambda[i:i+n])$
        % \State $\ell_{\text{best}} = \ell(\theta)$
        \ForAll{$\theta' \in \{0,\pi/4,\pi/2\}^{|\Lambda'|}$}
            \State $\theta_{\text{curr}} = \theta' \cup \theta_{\text{best}}[\Lambda \backslash \Lambda']$
            % \State $\ell_{\text{current}} = \ell(\theta_{\text{current}})$
            \If{$\ell(\theta_{\text{curr}}) < \ell(\theta_{\text{best}})$ or $\text{random}(0,1) < \varepsilon$}
                \State $\theta_{\text{best}} = \theta_{\text{curr}}$
                % \State $\ell_{\text{best}} = \ell_{\text{current}}$
            \EndIf
        \EndFor
    \EndFor
    \State \textbf{return} $\theta_{\text{best}}$
\EndFunction
\end{algorithmic}
\end{algorithm}
    
Algorithm~\ref{prog:greedyOptimization} uses the following notation. The problem-specific input variables are:
\begin{itemize}
    \item $\ell(\theta)$: Loss function as a function of measurement angles $\theta$.
    \item $\Lambda$: A list of slices compatible with the temporal order of the ansatz. For the $(2,0)$-MuTA layer of Fig.~\ref{fig:oneLayerMuta}, we set $\Lambda = [\{0\}, \{5\}, \{6\}, \{1\}, \{2,7\}, \{3,8\}, \{4,9\} ]$.
\end{itemize}

The hyperparameters are:
\begin{itemize}
    \item $\varepsilon$: Probability of exploring a random set of angles.
    \item $n_{\text{reset}}$: Maximum number of resets.
    \item $L_{\max}$: Maximum number of slices for simultaneous optimization.
    \item $\delta$: Threshold loss value for early stopping.
\end{itemize}

\newpage

\end{document}